\documentclass[11pt,letterpaper]{article}

\usepackage[margin=1in]{geometry}
\usepackage{theorem,latexsym,graphicx,amssymb}
\usepackage{amsmath}
\usepackage{float}
\usepackage{xspace}
\usepackage{paralist}
\usepackage{enumerate,multicol}
\usepackage{cases}
\usepackage{caption}
\usepackage{graphicx}
\usepackage{tikz,pgfmath}
\usetikzlibrary{matrix,positioning,quotes}
\usetikzlibrary{shapes.multipart}
\usetikzlibrary{calc}
\usetikzlibrary{arrows,decorations.markings}
\usetikzlibrary{matrix}
\usepackage{url}
\usepackage[ruled,linesnumbered]{algorithm2e}
\usepackage{hyperref}
\usepackage[noend]{algpseudocode}
\usepackage{xcolor}
\usepackage{multirow}
\usepackage{multicol}
\usepackage{graphicx}
\usepackage{subcaption}
\usepackage{varwidth}
\usepackage{wrapfig}
\usepackage{bbm}
\usepackage{enumitem}

\newenvironment{proof}{{\bf Proof:  }}{\hfill\rule{2mm}{2mm}\vspace*{5pt}}
\newenvironment{proofof}[1  ]{{\vspace*{5pt} \noindent\bf Proof of #1:  }}{\hfill\rule{2mm}{2mm}\vspace*{5pt}}
\numberwithin{figure}{section}
\numberwithin{equation}{section}
\newtheorem{definition}{Definition}[section]
\newtheorem{remark}{Remark}[section]

\newtheorem{theorem}{Theorem}[section]
\newtheorem{lemma}{Lemma}[section]
\newtheorem{claim}{Claim}[section]

\newtheorem{observation}{Observation}[section]

\SetKwInOut{Parameter}{Parameters}

\newcommand{\alg}{\textsf{ALG}}
\newcommand{\opt}{\textsf{OPT}}

\newcommand{\ue}{\textbf{Unbiased Estimator}}

\newcommand{\indic}{\mathbbm{1}}
\newcommand{\E}{\mathop{{}\mathrm{E}}}
\newcommand{\Var}{\mathrm{Var}}

\newcommand{\given}{\;\vert\;}
\newcommand{\nonfrac}{0.718}
\newcommand{\nonint}{0.666}
\newcommand{\nonhardness}{0.75}
\newcommand{\nj}{j}

\newcommand{\nonp}{\mathcal{P}}

\newcommand{\eqdef}{\overset{\mathrm{def}}{=\mathrel{\mkern-3mu}=}}

\newcommand{\vect}[1]{\ensuremath{\mathbf{#1}}}
\newcommand{\types}{\vect{t}}
\newcommand{\typesmi}[1][i]{\types_{\text{-}#1}}
\newenvironment{reminder}[1]{\bigskip
	\noindent {\bf Reminder of #1.  }\em}{\smallskip}
	
\newcommand{\hongxunfixed}[1]{\textcolor{green}{}}

\newcommand{\rtypes}{\vect{s}}

\title{(Fractional) Online Stochastic Matching via Fine-Grained Offline Statistics}

\author{Zhihao Gavin Tang\thanks{ITCS, Shanghai University of Finance and Economics. Email: \texttt{tang.zhihao@mail.shufe.edu.cn}}
\and Hongxun Wu\thanks{IIIS, Tsinghua University. Email: \texttt{wuhx18@mails.tsinghua.edu.cn}}
\and Jinzhao Wu\thanks{CFCS, Peking University. Email: \texttt{jinzhao.wu@pku.edu.cn}}}

\begin{document}

\maketitle

\begin{abstract}
Motivated by display advertising on the internet, the online stochastic matching problem is proposed by Feldman, Mehta, Mirrokni, and Muthukrishnan (FOCS 2009). Consider a stochastic bipartite graph with offline vertices on one side and with i.i.d. online vertices on the other side. The algorithm knows the offline vertices and the distribution of the online vertices in advance. Upon the arrival of each online vertex, its type is realized and the algorithm immediately and irrevocably decides how to match it. In the vertex-weighted version of the problem, each offline vertex is associated with a weight and the goal is to maximize the total weight of the matching.

In this paper, we generalize the model to allow non-identical online vertices and focus on the fractional version of the vertex-weighted stochastic matching. We design fractional algorithms that are $0.718$-competitive and $0.731$-competitive for non i.i.d. arrivals and i.i.d. arrivals respectively. We also prove that no fractional algorithm can achieve a competitive ratio better than $0.75$ for non i.i.d. arrivals.
Furthermore, we round our fractional algorithms by applying the recently developed multiway online correlated selection by Gao et al. (FOCS 2021) and achieve $0.666$-competitive and $0.704$-competitive integral algorithms for non i.i.d. arrivals and i.i.d. arrivals. 
Our results for non i.i.d. arrivals are the first algorithms beating the $1-1/e \approx 0.632$ barrier of the classical adversarial setting. Our $0.704$-competitive integral algorithm for i.i.d. arrivals slightly improves the state-of-the-art $0.701$-competitive ratio by Huang and Shu (STOC 2021). 
\end{abstract} 
\section{Introduction}
\label{sec:intro}

Since the seminal work of Karp, Vazirani, and Vazirani~\cite{stoc/KarpVV90}, online bipartite matching has been extensively studied in the online algorithms literature. In the classical setting, the vertices on one side of an underlying bipartite graph are known upfront to the algorithm. We refer to this set of vertices as offline vertices. The vertices on the other side of the graph, known as online vertices, are revealed in a sequence. Upon the arrival of each vertex, the algorithm makes an irrevocable matching decision between the vertex and its offline neighbors. The goal is to maximize the size of the matching produced by the algorithm. 
Karp et al. proposed the \textsf{Ranking} algorithm that achieves the optimal competitive ratio of $1-\frac{1}{e}$. Later, Aggarwal et al.~\cite{soda/AggarwalGKM11} generalized the algorithm to the vertex-weighted setting and attained the same competitive ratio.

Motivated by display advertising on the internet, Feldman et al.~\cite{focs/FeldmanMMM09} initiated the study of online stochastic matching. In the stochastic setting, each online vertex is drawn i.i.d. from a priori known distribution and the realization of each vertex is revealed on its arrival. Feldman et al.~\cite{focs/FeldmanMMM09} designed the first $0.67$-competitive algorithm that bypasses the $1-\frac{1}{e}$ impossibility result of the worst-case model, under an extra assumption of integral arrival rate. Later on, a line of subsequent works improved the competitive ratio to $0.729$~\cite{esa/BahmaniK10,mor/ManshadiGS12,mor/JailletL14,algorithmica/BrubachSSX20a} for unweighted graphs, and extended the results to vertex-weighted graphs~\cite{mor/JailletL14} and edge-weighted graphs~\cite{wine/HaeuplerMZ11,algorithmica/BrubachSSX20a}. 
Without the assumption of integral arrival rate, Manshadi, Oveis Gharan, and Saberi \cite{mor/ManshadiGS12} designed a $0.702$-competitive algorithm, and the ratio is improved to $0.706$ by Jaillet and Lu~\cite{mor/JailletL14} and to $0.711$ by Huang and Shu~\cite{stoc/HuangS21}. In the vertex-weighted online stochastic matching, the state-of-the-art competitive ratio is $0.701$ by Huang and Shu~\cite{stoc/HuangS21}. Manshadi, Oveis Gharan, and Saberi \cite{mor/ManshadiGS12} also proved that no algorithm can be better than $0.823$-competitive.

\subsection{Our Results}
\paragraph{Fractional Online Stochastic Matching.} In this paper, we study the \emph{fractional} version of online stochastic matching that allows us to fractionally match between a pair of vertices, as long as the total matched portion of each vertex is at most one. Furthermore, we generalize the online stochastic matching problem to \emph{non-identical} vertices. That is, the type of each vertex is drawn independently while the distributions of all vertices can be distinct.
 
We propose $0.731$-competitive and $\nonfrac$-competitive fractional algorithms for online stochastic matching with i.i.d. arrivals and non i.i.d. arrivals, respectively. Our algorithms and analysis automatically apply to the \emph{vertex-weighted} version without any adjustment.
We also provide a $\nonhardness$ hardness result for the fractional online stochastic matching problem with non i.i.d. arrivals.

\begin{table}[h]
\centering
\begin{tabular}{|c|cc|ccc|}
\hline
\multirow{2}{*}{Arrival}    & \multicolumn{2}{c|}{Previous Results}                                      & \multicolumn{3}{c|}{Our Results (Vertex-weighted)}                                                                \\ \cline{2-6} 
                            & \multicolumn{1}{c|}{Integral}                     & Hardness               & \multicolumn{1}{c|}{Integral}               & \multicolumn{1}{c|}{Fractional}             & Hardness              \\ \hline
\multirow{2}{*}{I.I.D.}     & \multicolumn{1}{c|}{Unweighted: 0.711~\cite{stoc/HuangS21}}                   & \multirow{2}{*}{0.823~\cite{mor/ManshadiGS12}} & \multicolumn{1}{c|}{\multirow{2}{*}{0.704}} & \multicolumn{1}{c|}{\multirow{2}{*}{0.731}} & \multirow{2}{*}{-}    \\
                            & \multicolumn{1}{c|}{Vertex-weighted: 0.701~\cite{stoc/HuangS21}}              &                        & \multicolumn{1}{c|}{}                       & \multicolumn{1}{c|}{}                       &                       \\ \hline
\multirow{2}{*}{Non I.I.D.} & \multicolumn{1}{c|}{\multirow{2}{*}{$1-1/e \approx 0.632$~\cite{stoc/KarpVV90,soda/AggarwalGKM11}}} & \multirow{2}{*}{-}     & \multicolumn{1}{c|}{\multirow{2}{*}{\nonint}} & \multicolumn{1}{c|}{\multirow{2}{*}{\nonfrac}} & \multirow{2}{*}{0.75} \\
                            & \multicolumn{1}{c|}{}                             &                        & \multicolumn{1}{c|}{}                       & \multicolumn{1}{c|}{}                       &                       \\ \hline
\end{tabular}
\caption{A summary of the state-of-the-art results and our results.}
\end{table}

\paragraph{Online Rounding via Online Correlated Selection.} 
A general and widely-applicable technique in the online algorithms literature is a two-step approach that first design a fractional online algorithm and then apply an online rounding scheme to it. 
As implications of our fractional algorithms, we apply the recently developed online rounding tool, namely \emph{multiway online correlated selection}, as a blackbox to our fractional algorithm and achieve $0.704$-competitive and $0.666$-competitive integral algorithms for the vertex-weighted online stochastic matching problem with i.i.d. arrivals and non i.i.d. arrivals, respectively.
The online correlated selection technique was introduced in the breakthrough work of Fahrbach et al.~\cite{focs/FahrbachHTZ20} as a key subroutine for the edge-weighted online bipartite matching problem. Recently, Blanc and Charikar~\cite{corr/BlancC21}, and Gao et al.~\cite{corr/GaoHHNYZ21} extended the rounding scheme to multiway online correlated selection. For the convenience of our work, we apply the multiway online correlated selection of Gao et al.~\cite{corr/GaoHHNYZ21}.

For the non i.i.d. arrival setting, our fractional $\nonfrac$-competitive and integral $\nonint$-competitive algorithms are the first algorithms beating the $1-1/e$ barrier of the classical adversarial setting.\footnote{The Ranking algorithm achieves a competitive ratio of $1-1/e$ for adversarial arrivals and breaks this bound in the i.i.d. setting. However, it does not beat the $1-1/e$ barrier for non-i.i.d. arrivals when the instance is the (deterministic) upper-triangle graph for which Ranking is exactly $1-1/e$-competitive.} 

For the i.i.d. arrival setting, our fractional algorithm significantly improves the $0.711$ ratio of Huang et al.~\cite{stoc/HuangS21}, and is even slightly better than the state-of-the-art $0.729$ ratio of Jaillet and Lu~\cite{mor/JailletL14}, under the extra integral arrival rate assumption. Our integral algorithm slightly improves the state-of-the-art $0.701$ competitive ratio for vertex-weighted graphs.

\subsection{Our Techniques}
\paragraph{Beyond Two-choice Algorithms.} Prior to our work, the arguably most successful approach~\cite{mor/ManshadiGS12,mor/JailletL14,stoc/HuangS21} for attacking the online stochastic matching problem is a two-choice algorithm. Upon the arrival of each vertex, it makes at most two tries to match its neighbors, and if both attempts failed, the vertex is left unmatched even if there exists other unmatched neighbors. The restriction of two-choice facilitates the analysis, while at the same time, leads to obviously suboptimal behaviors of the algorithm.
We import the multiway online correlated selection as a rounding scheme that allows us to explore beyond two-choice algorithms.

\paragraph{Unbiased Estimators: Fine-Grained Offline Statistics.} Similar to previous works, our algorithms heavily rely on the offline statistics of the instance and indeed, we use more fine-grained offline statistics. All previous works explicitly or implicitly use the following offline statistics which we call independent estimators: $x_{u,v}(t_v) = \Pr[(u,v) \in \opt \mid t_v]$, i.e. the probability that edge $(u,v)$ is matched in the optimal matching conditioning on the realized type $t_v$ of the online vertex $v$. 
A first natural fractional algorithm that we study is to fractionally match each edge $(u,v)$ by $x_{u,v}(t_v)$ on the arrival of every $v$. 
In Section~\ref{sec:non-iid}, we fully characterize the performance of this algorithm in the non i.i.d. setting. Surprisingly, our analysis shows that this algorithm is the best among a much larger family of algorithms in the non i.i.d. setting.

Furthermore, we consider more fine-grained offline statistics, e.g.
\[
x_{u,v_j}(t_{v_1},t_{v_2},\ldots,t_{v_j}) = \Pr \left[ (u,v_j) \in \opt \mid t_{v_1},t_{v_2},\ldots,t_{v_j} \right],
\]
where $v_k$ denotes the $k$-th arriving vertex. This value is the probability that $(u,v_j)$ is matched in the optimal matching conditioning on the realized types $\{t_{v_k}\}_{k\le j}$ of all arrived vertices. By considering such statistics, our matching decision is inherently adaptive to the historical information of the instance. More generally, we can select an arbitrary subset of arrived vertices that includes the current vertex, and calculate the conditional probability with respect to the types of those vertices. We name such offline statistics as unbiased estimators since the expectation of such random variables $x_{u,v}$ equals to the offline probability of $(u,v)$ being matched in the optimal matching, for arbitrary selection of the subset.

Our fractional algorithm for the i.i.d. arrival setting in Section~\ref{sec:iid} is a mix of different unbiased estimators. To the best of our knowledge, we are the first to utilize such fine-grained statistics for the online stochastic matching problem.

\paragraph{Unbiased Estimator with Minimum Variance.} Instead of using the size of the fractional matching or the size of the rounded matching as a natural objective function, we provide a unified framework by reducing the maximum matching problems to the design of an unbiased estimator with minimum variance. Fixing an offline vertex, the expected fraction it receives is the same for all unbiased estimators. On the other hand, the total fraction received by each vertex is a random variable and in some of the realizations, it can be larger than $1$. This is wasteful since we cannot collect the part of the gain that exceeds $1$. In an ideal case, if the random variable is never greater than $1$, our fractional algorithm would then be optimal. Unfortunately, it is an impossible task to design an unbiased estimator with such properties, due to the online nature of the problem. Intuitively, if we can design an unbiased estimator with small variance, it should lead to a good performance. We establish the connection between the variance of an unbiased estimator to competitive ratio in Section~\ref{sec:min_variance}.

\subsection{Paper Organization}
We formally define the family of unbiased estimator algorithms in Section~\ref{sec:prelim}. In Section~\ref{sec:min_variance}, we establish a connection between the second moment of an unbiased estimator to its competitive ratio. In Section~\ref{sec:warm-up}, as a warm up, we provide a simple algorithm that breaks the $1-1/e$ barrier with non i.i.d. arrivals. In Section~\ref{sec:iid}, we present our main algorithm for the i.i.d. arrival setting. In Section~\ref{sec:non-iid}, we present our main algorithm for the non i.i.d. arrival setting and the hardness result.
\subsection{Related Work}

Simultaneous to our work, Huang, Shu, and Yan~\cite{corr/HuangSY22} apply the online correlated selection techniques to the i.i.d. online stochastic matching and also exploit the power of multiple-choice algorithms. They also study the edge-weighted setting with free disposal. The online correlated selection techniques have also been applied to the Adwords problem~\cite{focs/HuangZZ20} and the online bipartite matching problem with reusable resources~\cite{corr/DelongFNS21}.

Devanur et al.~\cite{jacm/DevanurJSW19} studied a similar non-identical setting for general online resource allocation problems. On the other hand, they focused on the case of small bids and their results are not directly comparable to ours. Closely related to the i.i.d. stochastic matching problem is the random arrival model. In this model, we assume that the online vertices of an underlying graph arrive in a random order. It is observed that any algorithm in the random arrival model would apply for the stochastic model. For unweighted graphs, Karande, Mehta, and Tripathi~\cite{stoc/KarandeMT11} and Mahdian and Yan~\cite{stoc/MahdianY11} proved that the Ranking algorithm is $0.696$-competitive. For vertex-weighted graphs, Huang et al.~\cite{talg/HuangTWZ19} and Jin and Williamson~\cite{wine/JinW20} achieved a $0.662$ competitive ratio.

There is an extensive study of fractional online matching algorithms in the adversarial setting, including $b$-matching~\cite{tcs/KalyanasundaramP00}, fractional edge-weighted online bipartite matching~\cite{wine/FeldmanKMMP09}, online matching with concave return~\cite{stoc/DevanurJ12}, fully online matching~\cite{jacm/HuangKTWZZ20,soda/HuangPTTWZ19, focs/HuangTWZ20}, online matching with general vertex arrival~\cite{icalp/WangW15}, adwords~\cite{focs/MehtaSVV05, esa/BuchbinderJN07}. On the other hand, the state-of-the-art integral algorithms are mostly not directly related to the fractional algorithms in these models. In other words, they do not admit a good online rounding scheme. Recently, Buchbinder, Naor, and Wajc~\cite{corr/BuchbinderNW21} revisit the classical online bipartite matching problem and study how much randomness is necessary to beat the $1/2$ barrier for deterministic algorithms, for which they apply the same two-step approach of rounding a fractional algorithm as us. 
 
\section{Preliminaries} 
\label{sec:prelim}

\paragraph{Online Stochastic Matching.} We study the online vertex-weighted stochastic matching problem with non i.i.d. vertex arrivals. Consider a bipartite graph with offline vertices on one side, and with stochastic online vertices on the other side. We use $L$ to denote the offline vertices and $R = \{v_j\}$ to denote the online vertices. Each offline vertex $u \in L$ is associated with a weight $w_u$. The type of each online vertex $v_j$ is drawn independently from a priori known distribution $D_j$. Here, each type specifies its edges incident to the offline vertices. The offline vertices together with their associated weights, and the distributions $D_j$'s of online vertices are known in advance. We use $T_j$ to denote the support of $D_j$. The type $t_j \in T_j$ of each $v_j$ is only realized on its arrival and the algorithm makes immediate and irrevocable matching decisions. The goal is to maximize the total weight of the matched offline vertices. We assume the arrival order of the vertices is unknown to the algorithm but is fixed in advance. We use $\types$ to denote the realized types $(t_1,t_2,\cdots,t_n)$ of all vertices and $\types_{\le j}, \types_{>j}$ to denote the realized types $(t_1, t_2, \cdots, t_j)$ of the first $j$ vertices and the types $(t_{j+1},\cdots, t_n)$ of the vertices after $v_j$. For every index set $I \subseteq [n]$, we use $\types_I$ to denote the types $(t_i)_{i\in I}$.

We shall also study the i.i.d. setting in which all $D_j$'s are the same. We make no assumptions on the arrival rate of each type.

\subsection{Fractional Online Stochastic Matching}

We focus on the design and analysis of fractional algorithms, and then apply the multiway online correlated selection by Gao et al.~\cite{corr/GaoHHNYZ21} as our rounding scheme.

On the arrival of each vertex $v_j$, based on the realized type $t_j$ of $v_j$, we construct a vector $\vect{x}_j \in [0,1]^{L}$ such that $x_{u,j} > 0$ only if $u$ is a neighbor of $v_j$ and $\sum_{u \in L} x_{u,j} \le 1$. Then we match $v_j$ to each $u \in L$ by a fraction of $x_{u,j}$.
For the ease of our presentation, we use $y_u = \sum_{j} x_{u,j}$ to denote the total fraction that $u$ received, and allow it to be greater than $1$. Then, the performance of our algorithm equals
\[
\alg = \sum_{u \in L} w_u \cdot \min (y_u, 1).
\]
The vector $\vect{x}_j$ is naturally decided based on the realized types of all arrived vertices. We abuse the notation $\vect{x}_j$ to denote a function that maps from types of vertices to (fractional) matching decisions.
We consider the following family of fractional algorithms that we call \emph{unbiased estimators}. 

\begin{definition}[Unbiased Estimator] $\vect{x}_j : T_1 \times T_2 \times \cdots \times T_j \to [0,1]^L$ is \emph{unbiased} if:
\begin{itemize}
	\item $\sum_{u \in L} x_{u,j}(\types_{\le j}) \le 1$ for all $\types_{\le j} \in T_1 \times T_2 \times \cdots \times T_j$;
	\item $x_{u,j}(\types_{\le j}) > 0$ only if the vertex $v_j$ with type $t_j$ has an edge to $u$. 
	\item $\E_{\types_{\le j}}[x_{u,j}(\types_{\le j})] = \Pr[(u,v_j) \in \opt ]$ for all $u \in L$.
\end{itemize}
Further, we say an fractional algorithm is an unbiased estimator, if for every step $j\in[n]$, the constructed vector $\vect{x}_j$ is unbiased.
\end{definition}
\begin{remark}
All the estimators $x_{u,j}$ used in the work shall be the probability that edge $(u,v_j)$ is in the optimal matching, conditioning on the types of a subset of online vertices. With a sample access to each online vertex, this family of statistics can be estimated by standard Monte-Carlo algorithm within arbitrary additive accuracy (with high probability). Hence, our algorithms can be implemented efficiently with only an $\epsilon$ loss in the competitive ratio. 
\end{remark}

\subsection{Rounding via Online Correlated Selection}
We import the multiway online correlated selection formulation by Gao et al.~\cite{corr/GaoHHNYZ21} and present their setting and result for the ease of our application.

\paragraph{Multiway Online Selection Problem.~\cite{corr/GaoHHNYZ21}}  
Considers a set of elements $U$ and a selection process that proceeds in $n$ rounds. Each round $j \in [n]$ is associated with a non-negative vector $\vect{x}_j$ with $\sum_{u \in U} x_{u,j} \le 1$. The vectors are unknown at the beginning and are revealed to a multiway online selection algorithm at the corresponding rounds. Let $U_j = \{u\in U: x_{u,j} > 0\}$ be the set of elements with positive masses in round $j$. Upon observing the vector $\vect{x}_j$ for round $j$, the algorithm selects an element from $U_j$. Let $y_u = \sum_{j \in [n]} x_{u,j}$ be the cumulative mass of each element $u$. 

\begin{theorem} [Theorem 6, \cite{corr/GaoHHNYZ21}] \label{thm:ocs}
There exists a multiway online correlated selection such that any element $u$ with accumulated mass $y_u$ is selected with probability at least 
\[
p(y_u) \eqdef 1 - \exp\left( - y_u - \frac{1}{2}\cdot y_u^2 - \left(\frac{4 - 2\sqrt{3}}{3}\right) \cdot y_u^3\right)~.
\]
\end{theorem}

We remark that the original paper requires $\sum_{u \in U}x_{u,j}=1$ for each $\vect{x}_j$. In order to apply the algorithm to arbitrary vector $\sum_{u \in U} x_{u,j} \le 1$, we introduce a dummy vertex $u_{o}$ that connects to every online vertex and for each $j$, let $x_{u_o,j} =1- \sum_{u \in U}x_{u,j}$.

We also observe the concavity of function $p(y)$, that shall be used later in the proof of Lemma~\ref{lem:noniidworstper} and Lemma~\ref{lem:split}.

\begin{lemma} \label{lem:concave}
The function $p(y) = 1 - \exp\left( - y - \frac{1}{2}\cdot y^2 - \left(\frac{4 - 2\sqrt{3}}{3}\right) \cdot y^3\right)$ is concave.
\end{lemma}
\begin{proof}
Let $c=\frac{4 - 2\sqrt{3}}{3}$, we calculate the second-order derivative of $p(y)$:
\begin{align*}
p'(y)& = \left( 1+y+3c y^2 \right)\cdot \exp\left( - y - \frac{1}{2}\cdot y^2 - \left(\frac{4 - 2\sqrt{3}}{3}\right) \cdot y^3\right); \\
p''(y) & = \left((1+6cy) - (1+y+3cy^2)^2\right) \cdot \exp\left( - y - \frac{1}{2}\cdot y^2 - \left(\frac{4 - 2\sqrt{3}}{3}\right) \cdot y^3\right) \\
& = \left( (6c-2) y - (1+6c)y^2 - 6cy^3 - 9c^2y^4\right) \cdot \exp\left( - y - \frac{1}{2}\cdot y^2 - \left(\frac{4 - 2\sqrt{3}}{3}\right) \cdot y^3\right) < 0,
\end{align*}
where the last inequality follows from the fact that $6c-2\approx -0.928 < 0$.
\end{proof}

We then apply this OCS rounding scheme as a black-box to our fractional algorithm. See the following pseudocode for a formal description of our integral algorithm.

\begin{algorithm}[H]
	\caption{Unbiased Estimator + Online Correlated Selection}
	\begin{algorithmic}
	\State On the arrival of vertex $v_j$,
	\State \qquad 1. Construct a vector $\vect{x}_j \in [0,1]^{L}$ according our \ue \ algorithm.
	\State \qquad 2. Select a vertex $u \in L$ by applying the multiway OCS w.r.t. $\vect{x}_j$.
	\State \qquad 3. Match $v_j$ to $u$ if $u$ is not matched yet.
	\end{algorithmic}
\end{algorithm}

\subsection{Competitive Ratio}
We provide a unified analysis for the fractional algorithms and their corresponding integral versions by rounding with OCS. 
We conclude a competitive ratio of $\Gamma$ of our algorithms by proving a stronger statement that for every offline vertex, its matched fraction (probability of being matched) is at least $\Gamma$ times the probability it is matched in the optimal matching. We remark that a standard competitive analysis only need to show that the total expected matching is at least $\Gamma$ times the optimal matching.

As a reward, our analysis works for vertex-weighted graphs without any loss in the competitive ratio. In contrast, even in the i.i.d. arrival setting, prior works~\cite{mor/JailletL14, stoc/HuangS21, algorithmica/BrubachSSX20a} need to do different analyses for the unweighted and vertex-weighted settings. 
Formally, we prove the following lemma, that allows us to study each offline vertex separately.

\begin{lemma} \label{lem:vertex-by-vertex}
	For any unbiased estimator, its competitive ratio for the (non-i.i.d. vertex-weighted) fractional online stochastic matching problem is at least
	\[
	\min_{u \in L} \frac{\E_{\types}\left[\min(y_u(\vect{t}),1)\right]}{\E_{\types}\left[y_u(\vect{t}) \right]}~,\text{where } y_u(\types) = \sum_{j \in [n]} x_{u,j}(\types_{\le j})~.
	\]
	Moreover, if we round the unbiased estimator with OCS, its competitive ratio is at least 
	\[
	\min_{u \in L} \frac{\E_{\types}\left[p(y_u(\vect{t}))\right]}{\E_{\types}\left[y_u(\vect{t}) \right]}~.
	\]
\end{lemma}
\begin{proof}
We directly write $x_{u,j}, y_u$ and omit the inputs $\types_{\le j}, \types$ for notation simplicity.
According to the definition of unbiased estimator, we have 
\begin{multline*}
\E[\opt] = \sum_{u \in L} w_u \cdot \Pr[u \text{ is matched in } \opt] \\
= \sum_{u \in L} w_u \cdot \sum_{j \in [n]} \Pr[(u,v_j) \in \opt]= \sum_{u \in L} w_u \cdot \sum_{j \in [n]} \E_{\types_{\le j}}[x_{u,j}] = \sum_{u \in L} w_u \cdot \E_{\types}[y_u]~.	
\end{multline*}

For our fractional algorithm, $\E_{\types}[\textsf{FRAC}] = \sum_{u \in L} w_u \cdot \E_{\types}[\min (y_u,1)]$. Therefore,
\[
\frac{\E[\textsf{FRAC}]}{\E[\opt]} = \frac{\sum_{u \in L} w_u \cdot \E_{\types} [\min(y_u,1)]}{\sum_{u \in L} w_u \cdot \E_{\types} [y_u]} \ge \min_{u \in L} \frac{\E_{\types}[\min(y_u,1)]}{\E_{\types}[y_u]}~.
\]
 
Next, we consider the integral version by rounding with OCS. Recall the second condition of the definition of unbiased estimator and the property of OCS, the candidate vertex $v_j$ from the second step of our algorithm must be a neighbor of $v_j$. In other words, our attempt for matching $(v_j,u)$ is valid. 
To evaluate the performance of our algorithm, we observe the two types of randomness involved. The first type of randomness comes from the nature of the stochastic environment, i.e. the types $\types$ are drawn from the product distribution $D_1 \times D_2 \times \cdots \times D_n$. The second type of randomness comes from the random selection of multiway OCS. Fix any type vector $\types$, the vectors $\vect{x}_j$'s constructed in the first step of our algorithm are fixed according to the first condition above. Applying Theorem~\ref{thm:ocs}, we know that each vertex $u \in L$ is matched with probability at least $p(y_u)$. Consequently, $\E[\alg(\types)] \ge \sum_{u \in L} w_u \cdot p(y_u).$ By taking the randomness over $\types$, we have
\[
\E[\alg] \ge \sum_{u \in L} w_u \cdot \E_{\types} [p(y_u)]~.
\]
We conclude the lemma by noticing that 
\[
\frac{\E[\alg]}{\E[\opt]} \ge \frac{\sum_{u \in L} w_u \cdot \E_{\types} [p(y_u)]}{\sum_{u \in L} w_u \cdot \E_{\types} [y_u]} \ge \min_{u \in L} \frac{\E_{\types}[p(y_u)]}{\E_{\types}[y_u]}~.
\]
\end{proof}

\section{Unbiased Estimator with Minimum Variance}
\label{sec:min_variance}

Equipped with Lemma~\ref{lem:vertex-by-vertex}, it suffice to design an unbiased estimator such that the values $\E[f(y_u)]/\E[y_u]$ are large simultaneously for all offline vertices $u \in L$, where
\[
f(y) = 
\begin{cases}
 	\min(y,1) & \text{for the fractional version} \\
	p(y) & \text{for the integral version}
\end{cases}
\]
Fix an arbitrary offline vertex $u$, recall that $\E[y_u] = \Pr[u \text{ is matched in }\opt]$ for any unbiased estimator. Hence, we are aiming for an unbiased estimator with the largest possible $\E[f(y_u)]$.
There is a rich space of unbiased estimators and in principle, different functions of $f(y)$ might result in different optimal unbiased estimators. 

On the other hand, observe that the two functions we care about, $\min(y,1)$ and $p(y)$, are both concave. This suggests us to design an unbiased estimator with \emph{minimum variance}. 
Intuitively speaking, for two random variables $y_1,y_2$ with the same mean, $\E[f(y_1)]$ is more likely to be larger than $\E[f(y_2)]$, if the variance $\Var[y_1]$ of $y_1$ is smaller than the variance $\Var[y_2]$ of $y_2$. 

This intuitive claim fails to hold for arbitrary random variables. But for the purpose of our algorithm design, it is sufficient to use $\Var[y]$ or $\E[y^2]$ as a proxy for the objective function $\E[f(y)]$. 

The next lemma formalize this idea and reduces our algorithm design problem to unbiased estimator design with minimum variance. We approximate the functions $f(y)=\min(y,1)$ and $f(y)=p(y)$ by quadratic functions and prove that $\E[f(y)]$ must be large, provided $\E[y^2]$ is small.

\begin{lemma}\label{lem:variance-LP}
Suppose a non-negative random variable $y$ satisfies that $\E[y] = \mu$ and $\E[y^2] \leq \gamma$. 
\begin{enumerate}[label=(\alph*)]
	\item If $\mu \in [0,1]$ and $\gamma = \mu + \frac{1}{2}\mu^2 $, then $\E[\min(y,1)] \geq 0.646 \cdot \E[y]$;
	\item If $\mu \in [0,1]$ and $\gamma = \mu + \frac{1}{2}\mu^2 $, then $\E[p(y)] \geq 0.634 \cdot \E[y]$;
	\item If $\mu \in [0,1]$ and $\gamma = 1.05771\mu + 0.231\mu^2$, then $\E[\min(y,1)] \geq 0.731 \cdot \E[y]$.
	\item If $\mu \in [0,1]$ and $\gamma = 1.05771\mu + 0.231\mu^2$, then $\E[p(y)] \geq 0.704 \cdot \E[y]$.
\end{enumerate}
\end{lemma}

The proof of the Lemma is deferred to Appendix~\ref{appendix:variance-LP}. In Section~\ref{sec:warm-up} (respectively Section~\ref{sec:iid}), we construct an unbiased estimator such that the condition $\gamma=\mu+\frac{1}{2}\mu^2$ holds (respectively the condition $\gamma = 1.05771\mu + 0.231\mu^2$) for all $y_u$'s simultaneously, in the non i.i.d. arrival setting (respectively in the i.i.d. arrival setting).

\section{Warm up: Breaking $1 - \frac{1}{e}$ for Non I.I.D. Arrivals}\label{sec:warm-up}

In this section, we study the non-identical setting and propose an algorithm with competitive ratio strictly better than $(1-1/e)$. 
We first introduce the following two unbiased estimators. 

\paragraph{Independent Estimator.} On the arrival of each vertex $v_j$, fix its type $t_j$ and calculate the probability that $(u,v_j)$ is matched in the optimal matching when the types of all other vertices are resampled. I.e.,
\[
x_{u,j}(t_j) = \Pr_{\vect{\tilde{t}}_{\text{-}j}}\left[(u,v_j) \in \opt \left(t_j,\vect{\tilde{t}}_{\text{-}j}\right) \right]~.
\]
\paragraph{Fully-Correlated Estimator.} On the arrival of each $v_j$, fix the types $\types_{\le j}$ of all previously arrived vertices and calculate the probability that $(u,v_j)$ is matched in the optimal matching when the types of the remaining vertices are resampled. I.e.,
\[
x_{u,j}(\types_{\le j}) = \Pr_{\vect{\tilde{t}}_{>j}} \left[(u,v_j) \in \opt \left(\types_{\le j}, \vect{\tilde{t}}_{>j} \right) \right]~.
\]

The independent estimator ignores the realized types of previous vertices and calculates the probability $(u,v_j)$ is matched in \opt~conditioning on the type of $v_j$ being $t_j$. This is considerably the most natural unbiased estimator and the idea of resampling all remaining vertices is standard in the prophet inequality literature and the online contention resolution scheme literature, e.g. the resampling process as in~\cite{sigecom/EzraFGT20}. The fully-correlated estimator makes use of the full history and calculates the probability $(u,v_j)$ is matched in \opt~conditioning on the types $\types_{\le j}$ of all arrived vertices. 
It is straightforward to check both estimators satisfy the conditions of unbiased estimator.

\paragraph{Even Mix of IE/FE.} Consider an even mix between the independent and fully-correlated estimators, i.e. 
\[
x_{u,j}(\types_{\le j}) = \frac{1}{2}\cdot\Pr_{\vect{\tilde{t}}_{\text{-}j}}\left[(u,v_j) \in \opt \left(t_j,\vect{\tilde{t}}_{\text{-}j}\right) \right] + \frac{1}{2}\cdot\Pr_{\vect{\tilde{t}}_{>j}} \left[(u,v_j) \in \opt \left(\types_{\le j}, \vect{\tilde{t}}_{>j} \right) \right].
\]

\begin{theorem} \label{thm:even-mixing}
An even mix between independent estimator and fully-correlated estimator is $0.646$-competitive for the non i.i.d. vertex-weighted fractional online stochastic matching. Furthermore, by rounding it with OCS, the algorithm is $0.634$-competitive.
\end{theorem}
\begin{proof}
Fix an arbitrary vertex $u \in L$. Next, we bound the second moment of $y_u$ of our algorithm by studying the independent estimator and fully-correlated estimator separately.
For notation simplicity, we omit the subscript $u$ in all $x_{u,j}$'s and $y_u$'s. We use $\hat{x}_j, \hat{y}$ (resp. $\tilde{x}_j, \tilde{y}$) to denote $x_{u,j}, y_u$ by applying independent estimation (resp. by applying fully-correlated estimation). Then, our algorithm uses $x_j = \frac{1}{2}(\hat{x}_j+\tilde{x}_j)$ and $y = \frac{1}{2}(\hat{y}+\tilde{y})$.

Observe that $\hat{x}_j$'s are independent, while $\tilde{x}_j$'s are not. However, since $\tilde{x}_{j}$ and $\tilde{x}_{k}$ share the same types $\types_{\le \min(j,k)}$, they are naturally negatively correlated. On the other hand, the variance of each $\tilde{x}_{k}$ is larger than each $\hat{x}_k$ as more randomness is realized in $\tilde{x}_k$. We formalize these two properties in the following lemma.

\begin{lemma}
\label{lem:neg}
Suppose $\E[\hat{y}] = \E[\tilde{y}] = \Pr[u \text{ is matched in }\opt] = \mu$. We have
\begin{enumerate}
	\item $\E[\hat{y}^2] \le \mu^2 + \sum_{j \in [n]} \E[\hat{x}_j^2]$
	\item $\E[\tilde{y}^2] \le 2 \mu - \sum_{j \in [n]} \E[\tilde{x}_j^2]$
\end{enumerate}
\end{lemma}
\begin{proof}
We study independent estimator first. By the independence among $\hat{x}_j$'s. We have
\[
 \E[\hat y^2] = \E[\hat{y}]^2 + \Var[\hat{y}] = \mu^2 + \sum_{j \in [n]}\Var[\hat{x}_j] \le \mu^2 + \sum_{j \in [n]} \E[\hat{x}_j^2].
\]
Next, we study the fully-correlated estimator.	
We first prove the following negative correlation between $\tilde{x}_j$ and $\sum_{k>j}\tilde{x}_k$.
Fixing $\types_{\leq j}$, we have

\begin{multline}
\label{eqn:neg}
\tilde x_{j} (\types_{\leq j}) + \E_{\types_{>j}}\left[ \sum_{k>j} \tilde x_{k}(\types_{\leq k}) \ \middle\vert \  \types_{\leq j}  \right] =  \Pr[(u,v_j) \in \opt \mid \types_{\le j}] + \sum_{k>j} \Pr[(u,v_k) \in \opt \mid \types_{\le j}] \\
\leq \Pr[u \in \opt \mid \types_{\le j}] \leq 1.
\end{multline}

Consequently, we have
\begin{align*}
\E[\tilde y^2] &= \E\left[\left( \sum_{j \in [n]} \tilde x_{j} (\types_{\leq j}) \right)^2 \right] = \sum_{j \in [n]} \E[\tilde x_{j}^2 (\types_{\leq j})] + 2 \sum_{j \in [n]} \E\left[\tilde x_{j} (\types_{\leq j}) \cdot  \sum_{k>j} \tilde x_{k} (\types_{\leq k}) \right] \\
&= \sum_{j \in [n]} \E[\tilde x_{j}^2 (\types_{\leq j})] + 2 \sum_{j \in [n]} \E_{\types_{\le j}}\left[\tilde x_{j} (\types_{\leq j}) \cdot \E_{\types_{>j}}\left[ \sum_{k>j} \tilde x_{k} (\types_{\leq k}) \ \middle\vert \ \types_{\le j} \right] \right]\\
&\leq \sum_{j \in [n]} \E[\tilde x_{j}^2 (\types_{\leq j})] + 2 \sum_{j \in [n]} \E\left[\tilde x_{j} (\types_{\leq j}) (1 - \tilde x_{j} (\types_{\leq j})) \right] = 2\mu - \sum_{j=1}^n \E[\tilde x_{j}^2 (\types_{\leq j})],
\end{align*}
where the inequality is by Equation~\eqref{eqn:neg}. We finish the proof of the lemma.
\end{proof}

\begin{lemma} \label{lem:less}
For all $j \in [n]$, $\E[\hat x_{j}^2] \leq \E[\tilde x_{j}^2].$
\end{lemma}
\begin{proof}
By definition, $\hat{x}_j$ only depends on type $t_j$.
\begin{align*}
& \hat x_{j}(t_j) = \Pr_{\tilde{\types}_{\text{-}j}}[(u,v_j) \in \opt(t_j,\tilde{\types}_{\text{-}j})], \\
& \tilde{x}_{j}(\types_{\le j}) = \Pr_{\tilde{\types}_{>j}}[(u,v_j) \in \opt(\types_{\le j},\tilde{\types}_{>j})].
\end{align*}
Fixing $t_j$, $\hat{x}_j$ is a constant and $\E_{\types_{<j}}[\tilde{x}_j \mid t_j] = \hat{x}_j(\types_j)$. Together with the convexity of the function $x^2$, this implies
\[
\hat x_{j}^2 = \E_{\types_{<j}}[\tilde{x}_{j}(\types_{\leq j}) \mid t_j]^2 \leq \E_{\types_{<j}}[\tilde{x}_{j}^2(\types_{\leq j}) \mid t_j].
\]
Taking the randomness over $t_j$ concludes the proof of the lemma.
\end{proof}

Finally, we have
\begin{align*}
\E[y^2] &= \E \left[ \left(\frac{\hat{y} + \tilde{y}}{2}\right)^2 \right] \le \frac{1}{2} \cdot \left(\E[\hat{y}^2] + \E[\tilde{y}^2] \right)\\
&\leq \frac{1}{2}\mu^2 + \mu + \frac{1}{2} \sum_{j\in[n]} \left( \E[\hat x^2_{j}] - \E[\tilde x_{j}^2] \right) \leq \frac{1}{2}\mu^2 + \mu.
\end{align*}
Here, the first inequality follows by AM-GM inequality; the second inequality is by Lemma~\ref{lem:neg}; and the third inequality is by Lemma~\ref{lem:less}.
Observe that the random variable $y$ satisfies the first and second condition of Lemma~\ref{lem:variance-LP}, we have that 
\[
\E[\min(y,1)] \ge 0.646 \cdot \E[y] \quad \text{and} \quad \E[p(y)] \ge 0.634 \cdot \E[y].
\]
 This holds for every vertex $u \in L$. We conclude the proof of the theorem by Lemma~\ref{lem:vertex-by-vertex}.

\end{proof}

\section{I.I.D. Arrivals} \label{sec:iid}
\newcommand\numberthis{\addtocounter{equation}{1}\tag{\theequation}}

In this section, we prove our main theorem for the i.i.d. arrival setting.

\begin{theorem} \label{thm:iid}
There exists a $0.731$-competitive unbiased estimator for the i.i.d. vertex-weighted fractional online stochastic matching. Furthermore, by rounding it with OCS, the algorithm is $0.704$-competitive.
\end{theorem}

In Section~\ref{subsec:design}, we first describe the space of all unbiased estimators we study and the intuitions we collect from the warm-up analysis. In Section~\ref{subsec:alg}, we provide our algorithm. The full analysis of our algorithm is provided in Section~\ref{subsec:proof-iid}.

\subsection{Space of Unbiased Estimators} \label{subsec:design}
We study the following family of estimators and their convex combinations.
\paragraph{Subset-Resampling Estimators.} Fix an arbitrary index set $I_j \subseteq [j]$ that contains $j$ itself (i.e. $j \in I_j$). Consider the following estimator:
\[
x_{u,j}^{I_j}(\types_{I_j}) = \Pr_{\types_{[n]\setminus I_j}}\left[ (u,v_j) \in \opt \mid \types_{I_j} \right],
\]
where $\types_{I_j}$ denotes the types $t_i$ of each $i \in I_j$. It is straightforward to check this is an unbiased estimator. 
Observe that the independent estimator corresponds to $I_j = \{j\}$ and the fully-correlated estimator corresponds to $I_j = [j]$. 

Recall that we are aiming for an unbiased estimator with minimum variance. Fix an arbitrary vertex $u\in L$, the objective is to minimize
\[
\E_{\types}\left[y_u^2\right] = \sum_{j=1}^n \E_{\types_{\leq j}}\left[x_{u,j}^2 \right] + 2\sum_{j = 1}^n \sum_{k=j+1}^n  \E_{\types_{\leq k}}\left[x_{u,j} \cdot x_{u,k} \right].
\]

Thus, we are aiming for two goals at the same time: 1) to minimize the variance of each $x_{u,j}$, and 2) to negatively correlate $x_{u,j}$ and $x_{u,k}$.

For the first goal, we have the following intuition that generalizes Lemma~\ref{lem:less}.
\begin{itemize}
\item For all $I \subseteq I'$, $\E \left[ \left(x_{u,j}^{I} \right) ^2 \right] \le \E \left[ \left(x_{u,j}^{I'}\right)^2 \right]$. 
This is because $x_{u,j}^{I}(\types_{I}) = \E_{\types_{I'\setminus I}} \left[ x_{u,j}^{I'} \ \middle\vert \ \types_{I} \right], \forall \types_{I}$.
An immediate implication of this observation is that the independent estimator gives the smallest $\left(x_{u,j} \right)^2$ among all choices of $I$. Moreover, we should consider using index sets with smaller sizes in order to minimize the first part.
\end{itemize}

For the second goal, we have the following intuition from Lemma~\ref{lem:neg}. Recall that in the analysis of fully-correlated estimator, we proved that $\sum_{k\ge j} x_{u,k}^{[k]} \le 1$ in Equation~\eqref{eqn:neg}. We interpret this property as a negative correlation between $x_{u,j}^{[j]}$ and $x_{u,k}^{[k]}$. Indeed, since the two random variables share the same types $\types_{\le j}$, the more likely $(u,v_j)$ is matched in the optimal matching, the less likely $(u,v_k)$ is matched in the optimal matching.
More generally, we observe this negative correlation is actually caused by $t_j$ alone. 
Indeed, we have that
\begin{itemize}
\item For any $I_j \subseteq [j], I_k \subseteq [k], j<k$, $\E \left[ x_{u,j}^{I_j} \cdot x_{u,k}^{I_k \cup \{j\}} \right] \le \E \left[ x_{u,j}^{I_j} \cdot x_{u,k}^{I_k} \right]$. This observation suggests a different design of the unbiased estimator by adding $j$ into $I_k$. In the extreme case when we have $j \in I_k$ for all $j<k$, the only choice is to apply the fully-correlated estimator.
\end{itemize}

Observe that the two goals are contradictory to each other. Therefore, we must quantify the trade-off between the two intuitions and design our unbiased estimator in a careful way. In the warm-up section, we simply apply an even mix between the two extreme estimators, the independent estimator and the fully-correlated estimator. Next, we provide an improved estimator for the i.i.d. arrival setting.

\subsection{Our Algorithm}
\label{subsec:alg}
We consider a restricted subfamily of subset-resampling estimators.
\paragraph{Windowed Estimators.} We use $x_{u,j}^{(r)}$ to denote $x_{u,j}^{[j-r+1,j]}$, which we call windowed estimators, as it fix the types $\types_{[j-r+1,j]}$ of the last $r$ arrived vertices and resample the remaining types. I.e.,
$$x^{(r)}_{u,j}\left(\types_{[j-r+1,j]}\right) = \Pr_{\types_{[n] \setminus [j-r+1,j]}}[(u,v_j) \in \opt \mid \types_{[j - r + 1, j]}].$$

\paragraph{Mix of WE.} We apply the correlated estimation $x_{u,j}(\types)$ as a linear combination of windowed estimators: for each $j \in [n]$, 
\[
x_{u,j}(\types_{\le j}) = \frac{\beta}{n} \sum_{r=1}^{j - 1} x^{(r)}_{u,j} + \left(1 - \frac{j - 1}{n} \beta \right) x^{(j)}_{u,j},
\]
where  $\beta=0.79$ is a optimized constant used for all $j$'s. 

The following lemma is the most technical of our analysis for i.i.d. arrivals. We will present its proof in Section~\ref{subsec:proof-iid}. Finally, applying the lemma together with Lemma~\ref{lem:vertex-by-vertex} and \ref{lem:variance-LP}, we conclude the proof of Theorem~\ref{thm:iid}.
\begin{lemma}
\label{lem:variance}
The above mix of windowed estimator satisfies that 
\[
\E[y_u^2] \le 1.05771 \cdot \E[y_u] + 0.231 \cdot \left( \E[y_u] \right)^2, \forall u \in L.
\]
\end{lemma}

\subsection{Proof of Lemma~\ref{lem:variance}}\label{subsec:proof-iid}

Fix an arbitrary offline vertex $u \in L$. For notation simplicity, we shall omit the subscript $u$ for all the notations defined in the above subsections, e.g. we use $x_j, x_j^{(r)}, y$ to denote $x_{u,j}, x_{u,j}^{(r)}, y_u$. Let $a^{(r)}_j = \begin{cases} \frac{\beta}{n} & r < j \\ 1 - \frac{j - 1}{n}\beta & r = j \end{cases}$ be the coefficients of we defined in our algorithm. Then, we have $x_{j} = \sum_{r=1}^j a^{(r)}_j x^{(r)}_{j}$. 
Let $\mu = \E[y]$. We shall prove that $\E[y^2] \le 1.05771 \mu + 0.231 \mu^2$.

Similar as Section~\ref{subsec:design}, we expand $\E[y^2]$ into two parts.
\begin{align*}
\E[y^2] &=  \sum_{i=1}^n \E[x_{j}^2] + 2 \sum_{j=1}^n \sum_{k=j+1}^n \E[x_{j} x_{k}]\\
&= \sum_{j=1}^n \sum_{r_1=1}^j \sum_{r_2 = 1}^j a^{(r_1)}_j a^{(r_2)}_j \E\left[x^{(r_1)}_{j}x^{(r_2)}_{j}\right] + 2  \sum_{j=1}^n \sum_{k=j + 1}^n \sum_{r_1 = 1}^j \sum_{r_2 = 1}^k a_j^{(r_1)} a_k^{(r_2)} \E\left[x^{(r_1)}_{j} x^{(r_2)}_{k}\right]
\end{align*}

\begin{itemize}
\item In \ref{subsec:P_l}, we first define a few useful objectives that are used in both parts.
\item The detailed caculation of the two parts is presented in \ref{subsec:first_part} and \ref{subsec:second_part} respectively. 
\item Finally, we prove Lemma \ref{lem:variance} in \ref{subsec:sumup}.
\end{itemize}

\subsubsection{Definition of $P_\ell(\rtypes)$} \label{subsec:P_l}
 
Recall that $S$ is the support of the type distribution $D$. The following definition will be used throughout our calculation. 
 
 \begin{definition} \label{def:PJ}
 For any set $J \subset [n]$ and vector $\vect{s} \in S^{|J|} $, define $P_J : S^{|J|} \mapsto [0,1]$ to be:
 $$P_J(\vect{s}) = \Pr \left[ \exists j\in J, (u, v_j) \in \opt \mid \types_{J} = \vect{s} \right].$$ 
\end{definition}

In a word, $P_J(\rtypes)$ is the probability that $u$ is matched to an online vertex in $J$ conditioning on the types of online vertices in $J$ being $\rtypes$. 

\begin{observation}	\label{obs:sameJ}
For any $J_1$, $J_2$, if $|J_1| = |J_2|$, $P_{J_1}(\vect{s}) = P_{J_2}(\vect{s})$. 
\end{observation}
\begin{proof}
This follows from the fact that all online vertices has the same type distribution. Therefore they are symmetric to the offline optimum. 	
\end{proof}

Hence we can simply denote it using the size of the set $\ell = |J|$ instead of the actual set $J$. This results in the following definition. 

\begin{definition} \label{def:Pl}
For $\ell \in [n]$ and vector $\vect{s} \in S^\ell$, we define $P_\ell : S^\ell \mapsto [0,1]$ to be  
$$P_\ell (\vect{s}) = P_J(\vect{s}), \quad \forall J \subseteq [n], |J| = \ell$$
\end{definition}
By Observation \ref{obs:sameJ}, this is well-defined. 
We compute its expectation that we will use later. 

\begin{lemma} \label{lem:E[pl]}
$\E_{\vect{s} \sim D^{\ell}}[P_\ell (\vect{s})] = \frac{\mu \ell}{n}$, where $\mu=\E_{\types \sim D^n}[y(\types)]$.
\end{lemma}
\begin{proof}
By the property of I.I.D. arrival, we know that
\begin{align*}
\E_{\vect{s} \sim D^{\ell}}[P_\ell(\vect{s})] &= \E_{\vect{s}}\left[\Pr_{\types}\left[ \exists j \in[\ell], (u, v_j) \in \opt \mid \types_{[\ell]} = \vect{s} \right]\right] \\
&= \Pr_{\types} \left[ \exists j \in[\ell], (u, v_j) \in \opt \right] \\ 
&= \frac{\ell}{n} \Pr[u \in \opt] \tag{since the vertices are i.i.d.}\\ 
&= \frac{\mu \ell}{n} \tag{since $y$ is an unbiased estimator}
\end{align*}

\end{proof}

\subsubsection{The First Part} \label{subsec:first_part}

Let $\rtypes_\ell$ be a type vector of length $\ell$ sampled from $D^{\ell}$. The first step is to express each term $\E\left[x^{(r_1)}_j x^{(r_2)}_j \right]$ in the first part using $\E\left[P_\ell^2(\rtypes_{\ell})\right]$.

\begin{lemma} \label{lem:x^2}
Suppose $1 \leq r_1 \leq r_2 \leq j$. Let $\ell = |[j - r_1 + 1, j] \cap [j-r_2+1,j]| = r_1$ be the length of the intersection of these two windows. We have 
$$\E_{\types}\left[x^{(r_1)}_{j} x^{(r_2)}_{j}\right] \leq \frac{1}{\ell} \cdot \E_{\rtypes_\ell}\left[P_\ell^2 (\rtypes_\ell)\right]~.$$
\end{lemma}
\begin{proof}
\begin{align*}
& \quad \E_{\types}\left[x^{(r_1)}_{j} x^{(r_2)}_{j}\right] \\
& = \E_{\types_{[j-r_2+1,j]}}\left[\Pr \left[ (u, v_j) \in \opt \mid \types_{[j-r_1+1,j]} \right] \cdot \Pr \left[ (u, v_j) \in \opt \mid \types_{[j-r_2+1,j]} \right] \right] \\
& = \E_{\types_{[j-r_1+1,j]}} \left[\Pr \left[ (u, v_j) \in \opt \mid \types_{[j-r_1+1,j]} \right] \cdot \E_{\types_{[j-r_2+1,j-r_1]}}\left[\Pr \left[ (u, v_j) \in \opt \mid \types_{[j-r_2+1,j]} \right] \right]\right] \\
& = \E_{\types_{[j-r_1+1,j]}}\left[ \Pr \left[ (u, v_j) \in \opt \mid \types_{[j-r_1+1,j]} \right] ^2\right] \numberthis \label{equ:part1-iid}
\end{align*}

Recall that all vertices are symmetric. For all $k \in [j-r_1+1,j]$, we have 
$$\E_{\types_{[j-r_1+1,j]}}\left[ \Pr\left[(u, v_j) \in \opt \mid \types_{[j-r_1+1,j]} \right] ^2\right]  = \E_{\types_{[j-r_1+1,j]}}\left[ \Pr \left[(u, v_k) \in \opt \mid \types_{[j-r_1+1,j]} \right] ^2\right] $$
As a result, 
\begin{align*}
\eqref{equ:part1-iid} &= \frac{1}{r_1} \E_{\types_{[j-r_1+1,j]}}\left[\sum_{k=j-r_1+1}^j \Pr \left[(u, v_k) \in \opt \mid \types_{[j-r_1+1,j]} \right]^2\right] \\
&\leq \frac{1}{r_1} \E_{\types_{[j-r_1+1,j]}}\left[\left(\sum_{k=j-r_1+1}^j \Pr \left[(u, v_k) \in \opt \mid \types_{[j-r_1+1,j]} \right] \right)^2\right] \\
&= \frac{1}{r_1}\E_{\types_{[j-r_1+1,j]}}\left[\Pr\left[\exists k \in [j-r_1+1,j], (u, v_k) \in \opt \mid \types_{[j-r_1+1,j]}\right]^2\right] \tag{since the events are disjoint} \\ 
&= \frac{1}{\ell}\E_{\rtypes_{\ell}} \left [P_\ell^2 (\rtypes_\ell) \right] \tag{by $\ell = r_1$, Def.~\ref{def:PJ} and \ref{def:Pl}}	
\end{align*}

\end{proof}

Now we are ready to compute the first part. 

\begin{lemma} \label{lem:square-terms}
$$\sum_{j=1}^n \sum_{r_1=1}^j \sum_{r_2 = 1}^j a^{(r_1)}_j a^{(r_2)}_j \E_{\types}\left[x^{(r_1)}_{j}x^{(r_2)}_{j}\right] \leq \sum_{\ell = 1}^n \frac{\E_{\rtypes_{\ell}}\left[P_{\ell}^2 (\rtypes_{\ell})\right]}{\ell} \left(1 + 2\frac{n - 2\ell}{n}\beta + \frac{3 \ell^2-2\ell n}{n^2} \beta^2\right) + o(1)$$
\end{lemma}
\begin{proof}
First, let $\ell = \min(r_1, r_2)$ and apply Lemma \ref{lem:x^2}. We have 

\begin{align*}
\sum_{j=1}^n \sum_{r_1=1}^j \sum_{r_2 = 1}^j a^{(r_1)}_j a^{(r_2)}_j \E_{\types}\left[x^{(r_1)}_{j}x^{(r_2)}_{j}\right] &\leq \sum_{j=1}^n \sum_{r_1=1}^j \sum_{r_2 = 1}^j a^{(r_1)}_j a^{(r_2)}_j \frac{1}{\ell} \E_{\rtypes_{\ell}}\left[P_{\ell}^2(\rtypes_{\ell})\right] \tag{by Lemma \ref{lem:x^2}}\\
&= \sum_{\ell = 1}^n \sum_{j=\ell}^n  \left(\left(a_j^{(\ell)}\right)^2 + 2 a_j^{(\ell)} \sum_{r=\ell + 1}^{j} a_j^{(r)} \right) \frac{\E_{\rtypes_{\ell}}\left[P_{\ell}^2(\rtypes_{\ell})\right]}{\ell} \tag{rearranging the order}
\end{align*}

Then we plug in the definition of $a^{(r)}_j$. 

\begin{align*}
&\sum_{\ell = 1}^n \sum_{j=\ell}^n  \left(\left(a_j^{(\ell)}\right)^2 + 2 a_j^{(\ell)} \sum_{r=\ell + 1}^{j} a_j^{(r)} \right) \frac{\E_{\rtypes_{\ell}}\left[P_{\ell}^2(\rtypes_{\ell})\right]}{\ell} \\
=&\sum_{\ell = 1}^n \left(\left(a_\ell^{(\ell)}\right)^2 +  \sum_{j=\ell + 1}^n  \left(a_j^{(\ell)}\right)^2 + 2 \sum_{j=\ell + 1}^n a_j^{(\ell)} \left( a_j^{(j)} + \sum_{r=\ell + 1}^{j - 1} a_j^{(r)} \right) \right)  \frac{\E_{\rtypes_{\ell}}\left[P_{\ell}^2(\rtypes_{\ell})\right]}{\ell} \\
= &\sum_{\ell = 1}^n \left(\left(1 - \frac{\ell - 1}{n} \beta\right)^2 +  \sum_{j=\ell + 1}^n  \left(\frac{\beta}{n}\right)^2 + 2 \sum_{j=\ell + 1}^n \frac{\beta}{n}\left(1 - \frac{\ell}{n} \beta \right) \right)  \frac{\E_{\rtypes_{\ell}}\left[P_{\ell}^2(\rtypes_{\ell})\right]}{\ell} \numberthis \label{equ:part1} 
\end{align*}

Finally, we compute these summations and get the following.

\begin{align*}
(\ref{equ:part1}) = &\sum_{\ell = 1}^n \left(1 + \frac{2n - 4\ell + O(1)}{n} \beta + \frac{3 \ell^2 - 2 n \ell + O(1)(n + \ell)}{n^2} \beta^2 \right)  \frac{\E_{\rtypes_{\ell}}\left[P_{\ell}^2(\rtypes_{\ell})\right]}{\ell} \\
= &\sum_{\ell = 1}^n \left(1 + \frac{2n - 4\ell}{n} \beta + \frac{3 \ell^2 - 2 n \ell}{n^2} \beta^2 \right)  \frac{\E_{\rtypes_{\ell}}\left[P_{\ell}^2(\rtypes_{\ell})\right]}{\ell} + o(1),
\end{align*}

where the last step follows from the following calculation. First, from Lemma \ref{lem:E[pl]}, we know that $$\E_{\rtypes_{\ell}}[P_{\ell}^2(\rtypes_{\ell})] \leq \E_{\rtypes_{\ell}}[P_{\ell}(\rtypes_{\ell})] = \frac{\mu \ell}{n}.$$ 
Thus as $n$ goes to infinity, $$\sum_{\ell=1}^n \left(\frac{O(1)}{n} + \frac{O(1)(n + \ell)}{n^2}\right) \cdot \frac{\E[P_\ell^2]}{\ell} \leq \mu \sum_{\ell=1}^n \left(\frac{O(1)}{n^2}+ \frac{O(1) (n + \ell)}{n^3} \right) = o(1)~.$$
\end{proof}

\subsubsection{The Second Part} \label{subsec:second_part}

Similar with the first part, we first express each term $\E\left[x^{(r_1)}_{i,j} x^{(r_2)}_{i,k}\right]$ using $\E_{\rtypes_\ell}\left[P_\ell(\rtypes_{\ell})\right]$ and $\E_{\rtypes_\ell}\left[P_\ell^2(\rtypes_{\ell})\right]$. 

\begin{lemma} \label{lem:xixj}
For $1 \leq r_1 \leq j, 1 \leq r_2 \leq k$ and $j < k$, let $\ell = |[j - r_1 + 1, j] \cap [k-r_2+1,k]|$ be the length of the intersection of these two windows. Then, 

$$\E_{\types}\left[x^{(r_1)}_{j} x^{(r_2)}_{k}\right] = \begin{cases}
\E_{\rtypes_{\ell}}\left[\frac{P_\ell(\rtypes_{\ell})(1-P_\ell(\rtypes_{\ell}))}{\ell (n - \ell)}\right] & \ell > 0\\
\frac{\mu^2}{n^2} & \ell = 0
\end{cases} $$
\end{lemma}
\begin{proof}
\begin{itemize}
	\item If $\ell = 0$, $\types_{[j-r_1+1,j]}$ is independent of $\types_{[k-r_2+1,k]}$. Therefore, 
\begin{align*}
& \quad \E_{\types}\left[x^{(r_1)}_{j} x^{(r_2)}_{k}\right]  \\
&= \E_{\types_{[j-r_1+1,j] \cup [k-r_2+1,k]}} \left[\Pr\left[(u, v_j) \in \opt \mid \types_{[j-r_1+1,j]}\right] \cdot \Pr\left[(u, v_k) \in \opt \mid \types_{[k-r_2+1,k]}\right]\right]\\
&= \E_{\types_{[j-r_1+1,j]}} \left[\Pr [(u, v_j) \in \opt \mid \types_{[j-r_1+1,j]}]\right] \cdot \E_{\types_{[k-r_2+1,k]}}\left[\Pr[(u, v_k) \in \opt \mid \types_{[k-r_2+1,k]}]\right] \\
&= \Pr_{\types}[(u, v_j) \in \opt] \cdot \Pr_{\types}[(u,v_k) \in \opt]. \numberthis \label{equ:part2-iid}
\end{align*}

Since we have i.i.d. arrival, $u$ is matched to each online vertex with the same probability in the optimal matching. 
\[
\eqref{equ:part2-iid} = \left(\frac{\Pr[u \in \opt]}{n}\right)^2 = \frac{\mu^2}{n^2}.	
\]

\item If $\ell > 0$, let $[a,b] = [j - r_1 + 1, j] \cap [k-r_2+1,k]$. Then $\ell = b - a + 1$. Moreover since $j < k$, we must have $j \in [a,b]$ and $k \not\in [a,b]$. In this case, 
\begin{align*}
& \quad \E_{\types}\left[x^{(r_1)}_{j} x^{(r_2)}_{k}\right] \\
& = \E_{\types_{[j-r_1+1,j] \cup [k-r_2+1,k]}} \left[\Pr[(u, v_j) \in \opt \mid \types_{[j-r_1+1,j]}] \cdot \Pr[(u, v_k) \in \opt \mid \types_{[k-r_2+1,k]}]\right] \\
& =\E_{\types_{[a,b]}} \bigg[ \E_{\types_{[j-r_1+1,j] \setminus [a,b]}}\left[\Pr[(u, v_j) \in \opt \mid \types_{[j-r_1+1,j]}]\right] \cdot \\
& \phantom{=\E_{\types_{[a,b]}} \bigg[ } \E_{\types_{[k-r_2+1,k] \setminus [a,b]}} \left[\Pr[(u, v_k) \in \opt \mid \types_{[k-r_2+1,k]}]\right] \bigg] \\
& = \E_{\types_{[a,b]}} \left[\Pr\left[(u, v_j) \in \opt \mid \types_{[a,b]}\right] \cdot \Pr \left[(u, v_k) \in \opt \mid \types_{[a,b]}\right]\right]. \numberthis \label{equ:part2-iid-2}
\end{align*}
Similar as the proof of Lemma \ref{lem:x^2}, we use the property of I.I.D. arrival. 
For all $j' \in [a,b], k' \not\in [a,b]$, we know that
\begin{align*}
&\E_{\types_{[a,b]}} \left[\Pr\left[(u, v_j) \in \opt \mid \types_{[a,b]}\right] \cdot \Pr\left[(u, v_k) \in \opt \mid \types_{[a,b]}\right]\right] \\ 
= &\E_{\types_{[a,b]}} \left[\Pr\left[(u, v_{j'}) \in \opt \mid \types_{[a,b]}\right] \cdot \Pr\left[(u, v_{k'}) \in \opt \mid \types_{[a,b]}\right]\right].	
\end{align*}

Hence,
\begin{align*}
& \quad (\ref{equ:part2-iid-2}) \\
& = \frac{1}{\ell (n - \ell)} \E_{\types_{[a,b]}} \left[\sum_{j' \in [a,b]} \sum_{k' \in [n] \setminus [a,b]}\Pr\left[(u, v_{j'}) \in \opt \mid \types_{[a,b]}\right] \cdot \Pr\left[(u, v_{k'}) \in \opt \mid \types_{[a,b]}\right]\right] \tag{by the fact that the vertices are i.i.d. and $\ell = b-a+1$} \\
& = \frac{1}{\ell (n - \ell)}  \E_{\types_{[a,b]}} \left[\left(\sum_{j' \in [a,b]} \Pr\left[(u, v_{j'}) \in \opt \mid \types_{[a,b]}\right]\right) \cdot \left(\sum_{k' \in [n] \setminus [a,b]} \Pr\left[(u, v_{k'}) \in \opt \mid \types_{[a,b]}\right]\right)\right] \\
& = \frac{1}{\ell (n - \ell)}\E_{\types_{[a,b]}} \left[\Pr[\exists j' \in [a,b], (u, v_{j'}) \in \opt \mid \types_{[a,b]}] \cdot \Pr[\exists k' \in [n] \setminus [a,b],  (u, v_{k'}) \in \opt \mid \types_{[a,b]}]\right] \tag{since the events are disjoint}\\ 
& = \frac{1}{\ell (n - \ell)} \E_{\types_{[a,b]}} \left[\Pr[\exists j' \in [a,b], (u, v_{j'}) \in \opt \mid \types_{[a,b]}] \cdot \left(1 - \Pr[\exists k' \in [a,b],  (u, v_{k'}) \in \opt \mid \types_{[a,b]}]\right)\right] \\
& = \E_{\rtypes_\ell} \left[\frac{P_\ell(\rtypes_{\ell})(1 - P_\ell(\rtypes_{\ell}))}{\ell (n - \ell)}\right].
\end{align*}
\end{itemize}

Now we finish the proof of both cases. 

\end{proof}

Then we are ready to compute the second part.

\begin{lemma} \label{lem:cross-terms}
\begin{align*}
&2\sum_{j=1}^n \sum_{k=j + 1}^n \sum_{r_1 = 1}^j \sum_{r_2 = 1}^k a_j^{(r_1)} a_k^{(r_2)} \E_{\types}\left[x^{(r_1)}_{j} x^{(r_2)}_{k}\right]  \\ \leq &\sum_{\ell=1}^n \frac{2 \E_{\rtypes_{\ell}}[P_{\ell}(\rtypes_{\ell}) (1 - P_{\ell}(\rtypes_{\ell}))]}{\ell} \left( 1 + \frac{n - 5 \ell}{2n} \beta + \frac{8 \ell^2 - \ell n - n^2}{6n^2} \beta^2 \right) + \frac{1}{3} \mu^2 \beta + o(1)
\end{align*}
\end{lemma}
\begin{proof}
Let $\ell = |[j - r_1+1,j] \cap [k - r_2 + 1,k]|$ and apply Lemma \ref{lem:xixj}.

\begin{itemize}
	\item When $\ell = 0$ (i.e. $[j - r_1+1,j] \cap [k - r_2 + 1,k] = \emptyset$), we know that $\E_{\types}\left[x^{(r_1)}_{i,j} x^{(r_2)}_{i,k}\right] = \frac{\mu^2}{n^2}$. We first count in this part of contribution. Note since $j < k$, $[j - r_1+1,j] \cap [k - r_2 + 1,k] = \emptyset$ if and only if $r_2 \leq k - j$. Therefore, 
\begin{align*}
&2  \sum_{j=1}^n \sum_{k=j + 1}^n \sum_{r_1 = 1}^j \sum_{r_2 = 1}^k a_j^{(r_1)} a_k^{(r_2)} \E_{\types}\left[x^{(r_1)}_{j} x^{(r_2)}_{k}\right] \indic \left[ [j - r_1+1,j] \cap [k - r_2 + 1,k] = \emptyset \right] \\
= & \frac{2 \mu^2}{n^2} \sum_{j=1}^n \sum_{k=j + 1}^n \sum_{r_1 = 1}^j \sum_{r_2 = 1}^k a_j^{(r_1)} a_k^{(r_2)} \indic \left[ [j - r_1+1,j] \cap [k - r_2 + 1,k] = \emptyset \right] \tag{by Lemma \ref{lem:xixj}} \\
=&\frac{2 \mu^2}{n^2} \sum_{j = 1}^n \left(\sum_{r_1=1}^j  a_j^{(r_1)}  \right ) \cdot \sum_{k = j + 1}^n \left(\sum_{r_2 = 1}^{k - j} a_k^{(r_2)}\right) \numberthis \label{equ:part2-1}
\end{align*}

Then we plug in the definition of $a_j^{(r)}$. Note $\sum_{r_1=1}^j a_j^{(r_1)} = 1$ and $\sum_{r_2=1}^{k - j} a_k^{(r_2)} = \sum_{r_2=1}^{k - j} \frac{\beta}{n} = \frac{k - j}{n} \beta$. (Note $j \geq 1$ so that $k - j < k$.) 

\begin{equation}
\label{equ:part2-p1}
\eqref{equ:part2-1} = \frac{2 \mu^2}{n^2} \sum_{j = 1}^n \sum_{k = j + 1}^n \frac{k - j}{n} \beta = \frac{1}{3} \mu^2 \beta + o(1) 
\end{equation}

\item When $\ell > 0$ (i.e. $[j - r_1+1,j] \cap [k - r_2 + 1,k] \neq \emptyset$), we know that $\E_{\types}\left[x^{(r_1)}_{j} x^{(r_2)}_{k}\right] = \E_{\rtypes_{\ell}}\left[\frac{P_\ell(\rtypes_{\ell})(1-P_\ell(\rtypes_{\ell}))}{\ell (n - \ell)}\right]$. Moreover, Since $j < k$ and $\ell = |[j - r_1+1,j] \cap [k - r_2 + 1,k]|$, there are following two cases.
\begin{itemize}
\item  Either $k - r_2 + 1 \leq j - r_1 + 1$. In this case, $\ell = r_1$ and $r_2 \geq k - (j - r_1) = k - j + \ell$. Also we have $j \geq r_1 = \ell$. 
\item  Or $k - r_2 + 1 > j - r_1 + 1$. In this case, $\ell = j - (k - r_2)$ (i.e. $r_2 = k - j + \ell$) and $r_1 \geq \ell + 1$. Also we have $j \geq r_1 \geq \ell + 1$. 
\end{itemize}

Hence, 
\begin{align*}
&2  \sum_{\ell=1}^n\sum_{j=1}^n \sum_{k=j + 1}^n \sum_{r_1 = 1}^j \sum_{r_2 = 1}^k a_j^{(r_1)} a_k^{(r_2)} \E_{\types}\left[x^{(r_1)}_{j} x^{(r_2)}_{k}\right] \indic\left[ |[j - r_1+1,j] \cap [k - r_2 + 1,k]|= \ell \right]\\
= &2  \sum_{\ell=1}^n\sum_{j=1}^n \sum_{k=j + 1}^n \sum_{r_1 = 1}^j \sum_{r_2 = 1}^k a_j^{(r_1)} a_k^{(r_2)} \E_{\rtypes_{\ell}}\left[\frac{P_\ell(\rtypes_{\ell})(1-P_\ell(\rtypes_{\ell}))}{\ell (n - \ell)}\right] \indic \left[ |[j - r_1+1,j] \cap [k - r_2 + 1,k]|= \ell \right] \tag{Lemma \ref{lem:xixj}} \\
= & \sum_{\ell = 1}^n \frac{2 \E_{\rtypes_{\ell}}[P_\ell(\rtypes_{\ell})(1-P_\ell(\rtypes_{\ell}))]}{\ell (n - \ell)} \left(\sum_{j=\ell}^n a_j^{(\ell)}\sum_{k=j+1}^{n} \sum_{r = k - j + \ell}^{k} a_k^{(r)} + \sum_{j=\ell + 1}^n \sum_{r = \ell + 1}^j a_j^{(r)} \sum_{k=j + 1}^n a_k^{(k - j + \ell)} \right) \numberthis \label{equ:part2-2}
\end{align*}
Then we plug in the definition of $a_j^{(r)}$ to each part. 
\begin{itemize}
\item We first plug it into the first term in (\ref{equ:part2-2}). 
\begin{itemize}
\item When $j = \ell$,
\begin{align}
a_{j}^{(\ell)} \sum_{k=j+1}^{n} \sum_{r=k-j+\ell}^k a_k^{(r)} = a_{\ell}^{(\ell)} \sum_{k=\ell + 1}^n a_k^{(k)} = \left(1 - \frac{\ell}{n} \beta \right) \sum_{k=\ell + 1}^n \left(1 - \frac{k}{n} \beta\right). \label{equ:part2-term1}
\end{align}
\item When $\ell < j \leq n$, since $\sum_{r=k-j+\ell}^k a_k^{(r)} = 1 - \frac{k - j+\ell}{n} \beta $, 
\begin{align*}
\sum_{j=\ell + 1}^n a_j^{(\ell)} \sum_{k=j+1}^{n} \sum_{r=k-j+\ell}^k a_k^{(r)} &= \sum_{j=\ell + 1}^n a_j^{(\ell)} \sum_{k=j+1}^{n} \left(1 - \frac{k - j+\ell}{n} \beta \right) \\
&= \sum_{j=\ell + 1}^n \frac{\beta}{n}  \sum_{k=j+1}^{n} \left(1 - \frac{k - j+\ell}{n} \beta \right) \numberthis  \label{equ:part2-term2}
\end{align*}
\end{itemize}
\item Then we plug it into the second term in (\ref{equ:part2-2}). Since $\sum_{r=\ell+1}^j a_j^{(r)} = 1 - \frac{\ell}{n} \beta$ and $\sum_{k=j+1}^n a_k^{(k-j+\ell)} = \sum_{k=j+1}^n \frac{\beta}{n} = \frac{n - j}{n} \beta$ (which is due to the fact that $k - j + \ell < k$ because $j \geq \ell + 1$), 
\begin{align*}
 \sum_{j=\ell + 1}^n \sum_{r = \ell + 1}^j a_j^{(r)} \sum_{k=j + 1}^n a_k^{(k - j + \ell)} = &\sum_{j=\ell + 1}^n \left(\sum_{r = \ell + 1}^j a_j^{(r)}\right) \left(\sum_{k=j + 1}^n a_k^{(k - j + \ell)}\right)  \\
 =&\sum_{j=\ell+1}^n \left(1 - \frac{\ell}{n} \beta \right) \frac{n - j}{n} \beta \numberthis \label{equ:part2-term3}
\end{align*}
\end{itemize}

\begin{align*}
\eqref{equ:part2-2} = &\sum_{\ell = 1}^n \frac{2 \E_{\rtypes_{\ell}}[P_\ell(\rtypes_{\ell})(1-P_\ell(\rtypes_{\ell}))]}{\ell (n - \ell)} \cdot \left( \eqref{equ:part2-term1} + \eqref{equ:part2-term2} + \eqref{equ:part2-term3} \right) \\
= &\sum_{\ell = 1}^n \frac{2 \E_{\rtypes_{\ell}}[P_\ell(\rtypes_{\ell})(1-P_\ell(\rtypes_{\ell}))]}{\ell (n - \ell)} \cdot \\
& \left( \left(1 - \frac{\ell}{n} \beta\right) \sum_{k=\ell + 1}^n \left(1 - \frac{k}{n} \beta \right) + \sum_{j=\ell + 1}^n \frac{\beta}{n}  \sum_{k=j+1}^{n} \left(1 - \frac{k - j+\ell}{n} \beta \right) + \sum_{j=\ell+1}^n \left(1 - \frac{\ell}{n} \beta \right) \frac{n - j}{n} \beta \right) \\
= & \sum_{\ell=1}^n \frac{2 \E_{\rtypes_{\ell}}[P_\ell(\rtypes_{\ell})(1-P_\ell(\rtypes_{\ell}))]}{\ell}  \left( 1 + \frac{n - 5 \ell + O(1)}{2n} \beta + \frac{8 \ell^2 - \ell n - n^2 + O(1)(\ell + n)}{6n^2} \beta^2 \right) \\
= & \sum_{\ell=1}^n \frac{2 \E_{\rtypes_{\ell}}[P_\ell(\rtypes_{\ell})(1-P_\ell(\rtypes_{\ell}))]}{\ell}  \left( 1 + \frac{n - 5 \ell}{2n} \beta + \frac{8 \ell^2 - \ell n - n^2}{6n^2} \beta^2 \right)  + o(1) \numberthis \label{equ:part2-p2}
\end{align*}

Here the last step follow from the fact that $\E_{\rtypes_{\ell}}[P_\ell(\rtypes_{\ell})(1-P_\ell(\rtypes_{\ell}))] \leq \E_{\rtypes_{\ell}}[P_\ell(\rtypes_{\ell})] = \frac{\mu \ell}{n}$ (by Lemma \ref{lem:E[pl]}) and $$\sum_{\ell=1}^n \frac{\mu \ell}{n}\cdot \frac{1}{\ell} \left(\frac{O(1)}{n} + \frac{O(1)(\ell + n)}{n^2}\right) \leq \mu \sum_{\ell = 1}^n  \left(\frac{O(1)}{n^2} + \frac{O(1)(\ell + n)}{n^3}\right) = o(1).$$
\end{itemize}

Summing up (\ref{equ:part2-p1}) and (\ref{equ:part2-p2}) concludes the proof. 
\end{proof}

\subsubsection{Summary} \label{subsec:sumup}
Now, we have all necessary pieces to conclude the proof of Lemma~\ref{lem:variance}. By Lemma \ref{lem:square-terms} and Lemma \ref{lem:cross-terms}, we have the following.

\begin{align*}
\E_{\types}[y^2] 
&\le \sum_{\ell = 1}^n \frac{\E_{\rtypes_{\ell}}[P^2_\ell(\rtypes_{\ell})]}{\ell} \left( -1 + \frac{n + \ell}{n} \beta + \frac{\ell^2 - 5 \ell n + n^2}{3 n^2} \beta^2 \right)  \\ 
&+ \sum_{\ell=1}^n \frac{2 \E_{\rtypes_{\ell}}[P_\ell(\rtypes_{\ell})]}{\ell} \left( 1 + \frac{n - 5 \ell}{2n} \beta + \frac{8 \ell^2 - \ell n - n^2}{6n^2} \beta^2 \right) + \frac{1}{3} \mu^2 \beta + o(1) \numberthis \label{equ:final}
\end{align*}
By Lemma \ref{lem:E[pl]}, we know $\E_{\rtypes_{\ell}}[P_\ell(\rtypes_{\ell})] = \frac{\mu \ell}{n}$. Hence $\E_{\rtypes_{\ell}}[P^2_\ell(\rtypes_{\ell})] \geq \E_{\rtypes_{\ell}}[P_\ell(\rtypes_{\ell})]^2 = \frac{\mu^2 \ell^2}{n^2}$.

Suppose $\ell = \alpha n$ where $0 \leq \alpha \leq 1$. Recall that $\beta = 0.79$, the coefficient of $\frac{\E_{\rtypes_{\ell}}[P^2_\ell(\rtypes_{\ell})]}{\ell}$ in (\ref{equ:final}) is 
\[
-1 + \frac{n + \ell}{n} \beta + \frac{\ell^2 - 5 \ell n + n^2}{3 n^2} \beta^2 = \frac{\beta^2}{3} \alpha^2+\left(\beta - \frac{5}{3} \beta^2\right) \alpha + \left(-1 + \beta + \frac{1}{3}\beta^2 \right) < 0.21 \alpha^2 - 0.25 \alpha < 0,
\]
for all $0 < \alpha \leq 1$. 
Finally, we plug in our lower bound for $\E_{\rtypes_{\ell}}[P^2_\ell(\rtypes_{\ell})]$ and conclude the proof by:
\begin{align*}
\E_{\types}[y^2] & \leq \sum_{\ell = 1}^n \frac{\mu^2 \ell}{n^2} \left( -1 + \frac{n + \ell}{n} \beta + \frac{\ell^2 - 5 \ell n + n^2}{3 n^2} \beta^2 \right)  \\ 
& + \sum_{\ell=1}^n \frac{2 \mu}{n} \left( 1 + \frac{n - 5 \ell}{2n} \beta + \frac{8 \ell^2 - \ell n - n^2}{6n^2} \beta^2 \right) + \frac{1}{3} \mu^2 \beta + o(1) \\
& \leq \left(-\frac{1}{2} + \frac{7}{6}\beta - \frac{11}{36} \beta^2\right)\mu^2 + \left(2 - \frac{3}{2} \beta + \frac{7}{18} \beta^2 \right) \mu + o(1)\\
& \leq 1.05771\mu + 0.231\mu^2~. \tag{$\beta = 0.79$}
\end{align*}
  
\section{Non I.I.D. Arrivals}
\label{sec:non-iid}

In this section, we focus on the independent estimator for non i.i.d. vertex arrivals and characterize the worst-case instance for it. 

\begin{theorem} \label{thm:noniid}
The independent estimator is $\nonfrac$-competitive for the non i.i.d. vertex-weighted fractional online stochastic matching problem.  By rounding it with OCS, the algorithm is $\nonint$-competitive. Moreover, the independent estimator achieves the best competitive ratio among all subset-resampling estimators.\footnote{See Section~\ref{subsec:design} for the definition of subset-resampling estimators.}.
\end{theorem}

The full proof of Theorem~\ref{thm:noniid} is deferred to Section~\ref{sec:noniid-full}. Moreover, we complement it with the following hardness result.

\begin{theorem} \label{thm:noniidhardness}
No algorithm achieves a competitive ratio better than $0.75$ for the unweighted fractional online stochastic matching problem.
\end{theorem}

We first introduce some notations and definitions in Section~\ref{subsec:def} and then provide a proof sketch of Theorem~\ref{thm:noniid} in Section~\ref{subsec:strategy}. After that, we present the proof of Theorem~\ref{thm:noniidhardness} in Section~\ref{sec:noniidhardness}. Finally, we show the full proof of Theorem~\ref{thm:noniid} in Section~\ref{sec:noniid-full} and present the details of our experiments in Section~\ref{sec:exper}.

\subsection{Selection Rules and Independent Estimator} \label{subsec:def}
Fix an offline vertex $u$. We are only interested in the matching status of $u$ in the optimal matching and we think of $\opt$ as a \textit{selection rule}. That is, given the type vector $\types = (t_1, t_2, \dots, t_n)$ of online vertices, $\opt$ selects at most one of $j \in [n]$ and matches $u$ to $v_j$. I.e.
$$\opt(\types) = \begin{cases}
	\perp & \text{if $u$ is unmatched in the optimal solution.} \\
	j & \text{if $u$ is matched to $v_j$ in optimal solution.} \\
\end{cases}$$

\paragraph{Selection Rule.}  Formally, we define a (possibly randomized) selection rule $r$ to be a function that maps a type vector to a distribution over $[n] \cup \{\perp\}$: 
$$r: T_1 \times T_2 \times \cdots \times T_n \mapsto \Delta\left( [n] \cup \{\perp\} \right),$$
where $T_i$ is the support of the distribution $D_i$ for every $i \in [n]$ and $\Delta\left( [n] \cup \{\perp\}\right)$ is the family of all distributions supported on $[n] \cup \{\perp\}$. For $j \in [n] \cup \{\perp\}$, we use $r_j(\types)$ to denote the probability mass of $j$. Specifically, $r_{\perp}(\types)$ is the probability of selecting nothing. 

We will crucially use the following special class of selection rules, called \textit{permutation rules}. Each permutation rule specifies a total order $\pi$ over a subset $Y$ of identity-type pairs $\{(j,t)\}_{j\in [n], t \in T_j}$, and then selects the vertex whose identity together with its realized type appears first according to $\pi$. Note that when none of the $(j,t_j)$'s belongs to $Y$, we select nothing and return $\perp$. See the following for a more formal definition.

\begin{definition}[Permutation Rule $r^{\pi}$]

Suppose $\pi$ is a permutation over $Y \subset \{(j,t) \ \vert \ t \in T_j\}$. Based on ${\pi}$, $r^{\pi}$ denotes the following deterministic selection rule.

\begin{algorithm}[H]
\label{alg:selectionfunction}
  \SetAlgoLined
  \For{$i \leftarrow 1 \text{ to } |\pi|$}{
    {Suppose $\pi_i = (j, t)$.\\}
    \If{ $t_j = t$}{
        {Select $j$ and return $r^{\pi}(\types) = j$.}
    }
}
{Select nothing and return $r^{\pi}(\types) = \perp$.}
  \caption{Permutation Rule $r^{\pi}(\types) : T_1 \times T_2 \times \cdots \times T_n \rightarrow [n] \cup \{\perp\}$}
\end{algorithm}
\end{definition}

Previously, we have defined independent estimator with respect to $\opt$. Now we generalize this definition to any general selection rule $r$.

\paragraph{Independent Estimator for Selection Rule $r$.} On the arrival of each vertex $v_j$, fix its type $t_j$ and calculate the probability that $j$ is selected by $r$ when the types of all other vertices are resampled. I.e., 
$$x^r_{u,j}(t_j) = \E_{\tilde{\types}_{-j}}\left[r_j(t_j, \tilde{\types}_{-j})\right].$$ 
Note $x^{\opt}_{u,j}(t_j)$ is exactly the independent estimator defined in Section \ref{sec:warm-up}. 
Accordingly, we denote the cumulative mass of each element $u$ as $y^r_u(\types) = \sum_{j \in [n]} x^r_{u,j}(t_j)$.

\subsection{Proof Sketch} \label{subsec:strategy}
By Lemma \ref{lem:vertex-by-vertex}, our algorithm is $\Gamma$-competitive if for all $u \in L$, $\E\left[f(y^{\opt}_u)\right]\geq \Gamma \cdot \E\left[y^{\opt}_u\right]$ (where $f(y) = \min(y, 1)$ or $p(y)$). 

As discussed above, we think of $\opt$ as a selection rule and strengthening the above statement:
for an arbitrary selection rule $r$ and $u \in L$, $\E\left[f(y^{r}_u)\right]\geq \Gamma \cdot \E\left[y^{r}_u\right]$ holds. Therefore, the competitive ratio $\Gamma$ is lower bounded by the following optimization problem:
\begin{align}
\inf_{D_1, D_2, \dots, D_n} \inf_{ r} \frac{\E[f(y_u^r)]}{\E[y_u^r]} \label{opt-prob}
\end{align}

Our proof contains the following two steps:
\begin{itemize}
\item We first solve the inner optimization. We prove that for any fixed mean $\mu = \E\left[y^{r}_u\right]$, for any concave function $f$, the value of $\E\left[f\left(y^{r}_u\right)\right]$ is minimized by permutation rules.\footnote{We remark that the statement does not hold for all fixed distributions $D_1,D_2,\ldots,D_n$. On the other hand, we show that by optimizing the distributions simultaneously, we can without loss of generality restrict ourselves to permutation rules. See Lemma~\ref{lem:noniidworstper} for the formal statement.}
\item Next, we solve the outer optimization. For any fixed mean $\mu = \E\left[y^{r}_u\right]$, we characterize the worst case arrival distributions that minimize $\E\left[f(y^{r}_u)\right]$ under permutation rules.
\end{itemize}
Putting them together, the performance of the worst-case selection rule on its worst-case arrival for independent estimator gives the optimal constant $\Gamma$ for the above optimization problem. Furthermore, we show that under the worst-case distribution for independent estimator, all subset-resampling estimators have the same behavior as independent estimator, that implies the optimality of independent estimators among all subset-resampling estimators. Finally, we use computer assistance to calculate the numerical value of $\Gamma$.

Below, we elaborate the two steps for solving the optimization problem in more detail.
\paragraph{Selection Rule.} We formulate the inner optimization over $r$ by an optimization problem $\mathsf{CO}$ with concave objective and (matroid) polytope constraints. By concavity, its optimum must be achieved at extreme points of the polytope. 
Next, we adapt the approach in \cite{gamlath2019beating} to prove that every extreme point of the polytope corresponds to a permutation $\pi$ over $Y \subset \{(j,t) \ \vert \ t \in S_j\}$.

\paragraph{Arrival Distributions.} Next, we fix an arbitrary offline vertex $u \in L$ and a permutation rule $r$, and optimize over all possible arrival distributions. We prove that a careful subdivision of an online vertex could decreases the value of $\E[f(y_u^r)]$ while preserving the mean $\mu = \E[y_u^r]$. Informally, the subdivision increases the variance of $y^r_u$ and by the concavity of $f$, the value of $\E[f(y^r_u)]$ decreases. By applying the subdivision process iteratively, the worst case is achieved in the limit when each online vertex is infinitesimal.

\subsection{Hardness Result for Non I.I.D. Arrival}\label{sec:noniidhardness}

In this subsection, we give an upper bound for the fractional online matching problem with non i.i.d. arrivals.

\begin{proofof}{Theorem~\ref{thm:noniidhardness}}
Consider the following instance with $2$ offline vertices and $2$ online vertices. Suppose offline vertices are $\{u_1, u_2\}$, and online vertices are $\{v_1, v_2\}$. The vertices $v_1, v_2$ arrive in order, and $v_1$ deterministically connects to $u_1$ and $u_2$. For $v_2$, with probability $0.5$, it will have an edge to $u_1$, otherwise it will connect to $u_2$. The vertex weight for each offline vertex is $1$ (i.e., the unweighted case). See Figure \ref{fig:hardinstance}.

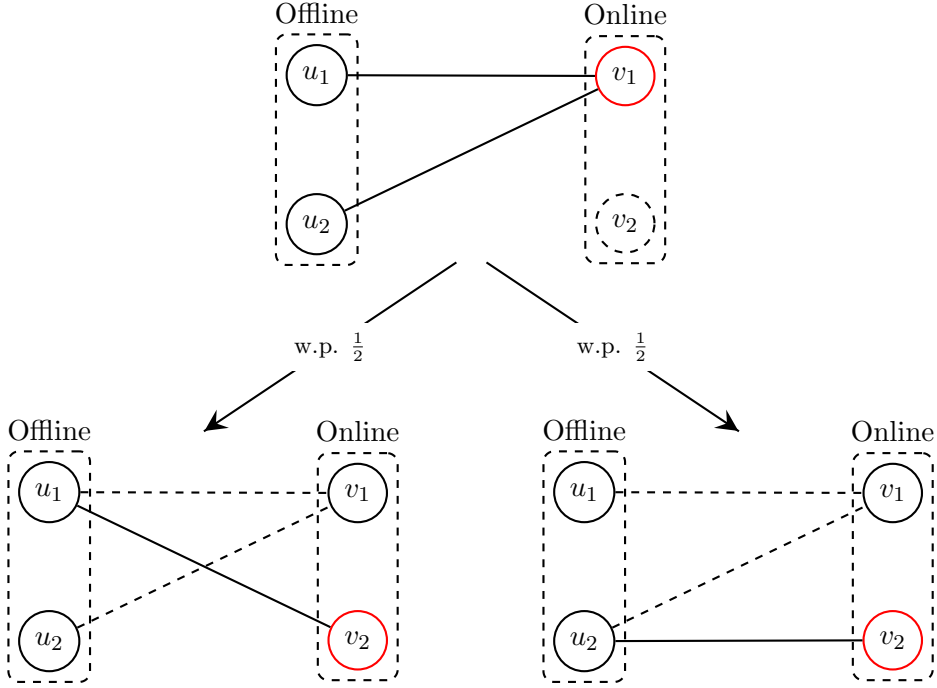
\begin{figure} 
\centering

\begin{tikzpicture}[thick,amat/.style={matrix of nodes,nodes in empty cells,
  fsnode/.style={draw,solid,circle,execute at begin node={$u_{\the\pgfmatrixcurrentrow}$}},
  ssnode/.style={draw,solid,circle,execute at begin node={$v_{\the\pgfmatrixcurrentrow}$}}}]

 \matrix (m1left) [amat,nodes=fsnode,label=above:Offline,row sep=3em,dashed,draw,rounded corners]  {
 \\
 \\ };

 \matrix (m1right) [amat,right=3cm of m1left,nodes=ssnode,label=above:Online,row sep=3em,dashed,draw,rounded corners]  {\node (m1-right) [draw=red,fill=none] {};\\
 \node (m1-right2) [dashed,fill=none] {};
 \\ };
\foreach \x in {1,...,2}
    {\draw  (m1left-\x-1) edge[thick] (m1-right); } 

\coordinate[right=1.5cm of m1left] (m1mida) ;
\coordinate[below=1.5cm of m1mida] (m1midb) ;
\coordinate[left=0.2cm of m1midb] (m1mid1) ;
\coordinate[right=0.2cm of m1midb] (m1mid2) ;

\coordinate[below=4cm of m1left] (TMP) ;

 \matrix (m2left)  [amat,left = 3cm of TMP, nodes=fsnode,label=above:Offline,row sep=3em,dashed,draw,rounded corners]  {
 \\
 \\ };

 \matrix(m2right) [amat,right=3cm of m2left,nodes=ssnode,label=above:Online,row sep=3em,dashed,draw,rounded corners]  
 {\\
 \node (m2right-2-1) [draw=red,fill=none] {};\\ };
\coordinate[right=1.5cm of m2left] (m2midb) ;
\coordinate[above=1.8cm of m2midb] (m2mid) ;

\foreach \x in {1,...,2}
    {\draw  (m2left-\x-1) edge[dashed, thick] (m2right-1-1); } 
 
 {\draw  (m2left-1-1) edge[thick] (m2right-2-1); } 

 \matrix (m3left)  [amat,right = 3cm of TMP, nodes=fsnode,label=above:Offline,row sep=3em,dashed,draw,rounded corners]  {
 \\
 \\ };

 \matrix(m3right) [amat,right=3cm of m3left,nodes=ssnode,label=above:Online,row sep=3em,dashed,draw,rounded corners]  {\\
\node (m3right-2-1) [draw=red,fill=none] {}; \\ };
 
\coordinate[right=1.5cm of m3left] (m3midb) ;
\coordinate[above=1.8cm of m3midb] (m3mid) ;

\foreach \x in {1,...,2}
    {\draw  (m3left-\x-1) edge[dashed, thick] (m3right-1-1); } 
 
 {\draw  (m3left-2-1) edge[thick] (m3right-2-1); } 
 
{\draw[decoration={markings,mark=at position 1 with
    {\arrow[scale=2,>=stealth]{>}}},postaction={decorate}] (m1mid1) -- (m2mid) node [midway,fill=white,anchor=center]{\footnotesize w.p. $\frac12$};}

{\draw[decoration={markings,mark=at position 1 with
    {\arrow[scale=2,>=stealth]{>}}},postaction={decorate}] (m1mid2) -- (m3mid) node [midway,fill=white,anchor=center]{\footnotesize w.p. $\frac12$};}

\end{tikzpicture}

\caption{The hard instance for fractional general arrival matching}
\label{fig:hardinstance}
\end{figure}

Notice that this instance will always have a perfect matching. Hence, the offline optimum equals $2$. 
However, any algorithm that matches $x_1$ fraction to $u_1$ and $x_2$ fraction to $u_2$ for the first online vertex $v_1$ will receive an expected total weight of at most $1.5$ after the realization of $v_2$, since
\[
    \E[\textsf{ALG}] \le \frac12 \left(x_1 + 1 \right) + \frac12 \left(x_2 + 1\right) \le \frac{3}{2},
\]   
where the first inequality follows from the fact that any vertex can match to at most a fraction of $1$, and the second inequality uses the fact that $x_1 + x_2\le 1$.
Therefore, no algorithm can achieve a competitive ratio better than $1.5/2 = 0.75$, that conclude the proof of the theorem.
\end{proofof} 
\subsection{Proof of Theorem~\ref{thm:noniid}}\label{sec:noniid-full}

\subsubsection{Permutation Rules} \label{subsec:worst-rule}

We prove that permutation rules are the worst among all selection rules for the inner optimization of \eqref{opt-prob} introduced in Section~\ref{subsec:strategy}. Formally, we prove Lemma~\ref{lem:noniidworstper} here.

\begin{definition}
For a distribution $D$ with support $T = \{a_1, a_2, \dots, a_n\}$, we say the distribution $D'$ with support $T' = S_1 \sqcup S_2 \sqcup \dots \sqcup S_n$\footnote{We use $\sqcup$ to denote a disjoint union of sets.} is a refinement of $D$ if and only if for all $i$, $\sum_{a \in S_i} \nonp^{D'}(a) = \nonp^D(a_j)$, where $\nonp^D, \nonp^{D'}$ are the probability mass functions of distributions $D,D'$ respectively. 

Moreover, for a sequence of distributions $\vect{D} = (D_1, D_2, \dots, D_n)$, we say $\vect{D'} = (D'_1, D'_2, \dots, D'_n)$ is its refinement if and only if each $D'_i$ is a refinement of $D_i$.
\end{definition}

\begin{lemma} \label{lem:noniidworstper}
Fix an arbitrary offline vertex $u \in L$ and the type distributions $\vect{D} = (D_1, D_2, \dots, D_n)$.
For any concave function $f$ that is non-negative when $x \geq 0$ and selection rule $r$ for $\vect{D}$ with $\E_{\types \sim \vect{D}}[f(y_u^{r})] = \mu$, there exists a refinement $\vect{D'}$ of $\vect{D}$ and a corresponding permutation rule $r^{\pi}$ with $\E_{\types \sim \vect{D'}}[f(y_u^{r^{\pi}})] = \mu$ such that $\E_{\types \sim \vect{D}}[f(y_u^{r}))] \geq \E_{\types \sim \vect{D'}}[f(y_u^{r^{\pi}})].$
\end{lemma}

For any fixed type $t$, we use $\nonp_j(t)$ to denote the probability mass of $t$ in $D_j$. We shall formalize an optimization problem $\mathsf{CO}$ to lower bound $\E_{\types \sim \vect{D}}[f(y_u^{r}))]$ for all selection rules $r$. To motivate this, we prove the following properties of selection rules.

For each selection rule $r$, let $z^r_j(t) \eqdef x^r_{u,j}(t) \cdot \nonp_j (t)$. Recall that $x_{u,j}^r(t)$ is the probability that selection rule matches $u,v_j$ together conditioning on the type of vertex $j$ to be $t$. By multiplying it with $\nonp_j(t)$, $z^r_j(t)$ is the probability that type $t$ is realized and selected by the selection rule $r$.

\begin{lemma} \label{lem:constraints_for_x}
For each selection rule $r$ with $\E[y^r_u] = \mu$, it satisfies the following conditions: 
\begin{align*}
	& \sum_{j\in[n]} \sum_{t\in T_j} z^r_j(t) = \mu & \\
	& \sum_{j\in[n]} \sum_{t\in S_j} z^r_j(t) \le g(S_1,S_2,\cdots,S_n) & \forall (S_1, S_2, \dots, S_n): S_i \subseteq T_i\\
	& z^r_j(t) \geq 0 & \forall j \in [n], t \in T_j
\end{align*}
where $g(S_1, S_2, \cdots, S_n) \eqdef 1 - \prod_{j\in [n]}\left(1 - \sum_{t \in S_j} \nonp_j(t)\right)$.
\end{lemma}
\begin{proof} 
The first equation follows from the assumption that $\E[y^r_u] = \mu$ and
\[
\E[y^r_u] = \sum_{j=1}^n \E_{t\sim D_j}[x^r_{u,j}(t)] = \sum_{j=1}^n \sum_{t \in T_j} x^r_{u,j}(t) \nonp_j(t).
\]
The second condition holds since 
\begin{multline*}
\sum_{j=1}^n \sum_{t\in S_j} z^r_j(t) = \sum_{j=1}^n \sum_{t \in S_j} x^r_{u,j}(t) \nonp_j(t) = \sum_{j=1}^n \sum_{t_j \in S_j} \E_{\types_{-j}}[r_j(t_j, \types_{-j})] \cdot \nonp_j(t_j) \\
= \sum_{\types} \Pr[\types \text{ is realized}] \cdot \sum_{j: t_j \in S_j} r_j(\types) \le \Pr_{\types}[\exists j \in [n], \text{$t_j \in S_j$}] = g(S_1, S_2, \dots, S_n),
\end{multline*}
where the inequality follows from the simple observation that for each type vector $\types$, the selection $r$ selects at most one of the $j \in [n]$. The last equality follows from the definition of function $g$.
\end{proof}

Now, consider the optimization problem $\mathsf{CO}$ below, where $g(S_1,S_2,\dots,S_n)$ is defined as in Lemma \ref{lem:constraints_for_x}. Let $\nonp(\types) \eqdef \nonp_1(t_1)  \nonp_2(t_2) \cdots \nonp_n(t_n)$ be the probability that types $\types$ are realized.
\begin{align}
\underset{\{z_j(t_j)\}}{\text{minimize:}} \quad &\sum_{\types} \nonp(\types) \cdot f\left( \sum_{j\in [n]}  z_j(t_j)  /\nonp_j(t_j)\right)\tag{$\mathsf{CO}$} \\
\text{subject to:} \quad &\sum_{j\in[n]} \sum_{t\in T_j} z_j(t)  = \mu \label{eq:NoniidPoly2} \\
\label{eq:NoniidPoly1} & \sum_{j\in[n]} \sum_{t\in S_j} z_j(t) \le g(S_1,S_2,\dots,S_n) & \forall (S_1, S_2, \dots, S_n): S_i\subseteq T_i \\
&z_j(t)\geq 0 &\forall j \in [n], t \in T_j \nonumber
\end{align}

By Lemma~\ref{lem:constraints_for_x}, we have that for every selection rule $r$ and every offline vertex $u \in L$, $\E_{\types \sim D}[f(y^r_u)] \geq \opt_{\mathsf{CO}}$, where $\opt_{\mathsf{CO}}$ is the optimal solution to the above optimization problem.
Next, we adapt the following lemma from \cite{gamlath2019beating} to show that the optimal solution of $\mathsf{CO}$ corresponds to a permutation $\pi$. We define the following notations. 
For a set sequences $\vect{A} = (A_1, A_2, \dots, A_n)$, we interchangeably view it as a set of index-type pairs $\vect{A} = \{(i,t) \mid t \in A_i \}$.
Consequently, we use 1) $\vect{A} \subseteq \vect{B}$ to denote that $i \in [n]$, $A_i \subseteq B_i$, 2) $\vect{A} \setminus \vect{B}$ to denote the set $\{(i, t) \mid t \in A_i \setminus B_i\}$, and 3) $\vect{A} \cup \vect{B}$ (and $\vect{A} \cap \vect{B}$) to denote the set sequence $(A_1 \cup B_1, \ldots, A_n \cup B_n)$ (and $(A_1 \cap B_1, \ldots, A_n \cap B_n)$).

\begin{lemma} [Lemma 12 of \cite{gamlath2019beating}, adapted] \label{lem:polytope}
For any concave function $f(y)$, let the optimal solution of $\mathsf{CO}$ be $\{z_j^*(t)\}$ and $Y = \{(j,t) \mid z_j^*(t) > 0 \}$. 
Then, there exists a sequence of set sequences $\vect{\tilde S}_1 \subsetneq \vect{\tilde S}_2 \subsetneq \cdots \subsetneq \vect{\tilde S}_{|Y|} \subseteq \vect{T} = (T_1, T_2, \dots, T_n)$ such that:
\begin{itemize}
\item For $1 \le k \le |Y| - 1$, the constraint \eqref{eq:NoniidPoly1} corresponding to $\vect{\tilde S_k}$ is tight. I.e., suppose $\vect{\tilde S_k} = (S_{k,1}, S_{k,2}, \dots, S_{k, n})$ where $S_{k, i}\subseteq T_i$,  then
$$\sum_{j\in[n]} \sum_{t\in S_{k,j}} z_j^*(t) = g \left( S_{k,1},S_{k,2},\cdots,S_{k,n} \right)$$
\item For $k = |Y|$, we have 
	$$\sum_{j\in [n]} \sum_{t\in S_{|Y|,j}} z_j^*(t) = \min \left( \mu, g \left( S_{|Y|,1},S_{|Y|,2},\cdots,S_{|Y|,n} \right) \right)$$
\end{itemize}
Moreover, $|\vect{\tilde S}_k \setminus \vect{\tilde S}_{k - 1}| = 1$ for every $1\le k \le |Y|$, where $\vect{\tilde S}_0=\emptyset$; and $\vect{\tilde S}_{|Y|}=Y$ corresponds to the set of non-zero variables.
\end{lemma}
\begin{proof}
Without loss of generality, we assume that $0 < \nonp(t) < 1$ for all types $t$. Then the function $g(S_1, S_2, \dots, S_n)$ is strictly submodular (c.f. Lemma 11 of \cite{gamlath2019beating}). 

Note that the objective function of $\mathsf{CO}$ is a composition of the concave function $f(y)$ and a linear function of the variables $\{z_j(t_j)\}_{j \in [n], t_j \in T_j}$. Thus, the objective is concave and the optimum value must be attained at one of the extreme points of the polytope defined by the constraints. Thus, it suffices to show that all extreme points of the polytope satisfy our claim.

For extreme point $\{z_j(t)\}$ of the polytope, let $Y = \{(j, t)\given z_j(t) > 0\}$ be the set of non-zero variables.  Since it is an extreme point, at least $|Y|$ constraints of \eqref{eq:NoniidPoly2} and \eqref{eq:NoniidPoly1} are tight. Let $\vect{A} = (A_1, A_2, \dots, A_n), \vect{B} = (B_1, B_2, \dots, B_n)$ be two set sequences whose corresponding constraints of \eqref{eq:NoniidPoly1} are tight. We have:
\begin{multline}
\label{eq:StrictModular}
g(\vect{A}) + g(\vect{B}) = \sum_{j \in [n]} \sum_{t \in A_j} z_j(t) + \sum_{j \in [n]} \sum_{t\in B_j} z_j(t)\\
= \sum_{j \in [n]} \sum_{t \in A_j\cap B_j} z_j(t) + \sum_{j \in [n]} \sum_{t\in A_j \cup B_j} z_j(t) \le g(\vect{A}\cap \vect{B}) + g(\vect{A} \cup \vect{B})
\end{multline}
where the last inequality follows from constraint \eqref{eq:NoniidPoly1}. However, if $\vect{A}\not\subseteq \vect{B}$ and $\vect{B} \not\subseteq \vect{A}$, inequality \eqref{eq:StrictModular} will contradict the strict submodularity of function $g$. Thus, for any $\vect{A}, \vect{B}$ whose corresponding constraints of \eqref{eq:NoniidPoly1} are tight, either $\vect{A}\subseteq \vect{B}$ or $\vect{B} \subseteq \vect{A}$.

Consequently, if $|Y|$ constraints  \eqref{eq:NoniidPoly1} are tight, then there must be $|Y|$ corresponding sets $\vect{\tilde S}_1 \subsetneq \vect{\tilde S}_2\subsetneq \dots \subsetneq \vect{\tilde S}_{|Y|}$ and for all $1 \leq i \leq |Y|$, $|\vect{\tilde S}_{i} \setminus \vect{\tilde S}_{i - 1}| = 1$. 

We are left with the case when $|Y| - 1$ constraints in \eqref{eq:NoniidPoly1} are tight and \eqref{eq:NoniidPoly2} is tight. Let $\vect{\tilde S}_1 \subsetneq \vect{\tilde S}_2\subsetneq \dots \subsetneq \vect{\tilde S}_{|Y| - 1}$ be the $|Y|-1$ set sequences whose corresponding constraints of \eqref{eq:NoniidPoly1} are tight. In this case, we let $\vect{\tilde S}_{|Y|} = Y$ and verify that $\vect{\tilde S}_{|Y| - 1} \subsetneq \vect{\tilde S}_{|Y|}$. For the sake of contradiction, suppose $\vect{\tilde S}_{|Y| - 1} = \vect{\tilde S}_{|Y|} = Y$. Then we must have $g(\vect{\tilde S}_{|Y| - 1}) = \mu$ because $\vect{\tilde S}_{|Y| - 1} = Y$ constains all nonzero variables and the corresponding constraint of \eqref{eq:NoniidPoly1} is tight. The constraint \eqref{eq:NoniidPoly2} is therefore dominated by $\sum_{j \in [n]} \sum_{t \in S_{|Y| - 1, j}} z_j(t) \leq g(\vect{\tilde S}_{|Y| - 1})$ in \eqref{eq:NoniidPoly1}. This implies that \eqref{eq:NoniidPoly2} can be removed, and there must be $|Y|$ tight constraints in \eqref{eq:NoniidPoly1} contradicting our assumption of having only $|Y| - 1$ tight constraints in \eqref{eq:NoniidPoly1}. So we must have that $\tilde{S}_{|Y| - 1}$ is a strict subset of $\tilde S_{|Y|}$.

Consequently, we also have $|\tilde S_{i} \setminus \tilde S_{i - 1}| = 1$ for all $1 \leq i \leq |Y|$. This finishes our proof. 
\end{proof}

Now, we are ready to conclude the proof of our main lemma.

\begin{proofof}{Lemma~\ref{lem:noniidworstper}}
Given distributions $\vect{D}$ and $u \in L$, we solve the corresponding optimization problem and let $\{z_j^*(t)\}$ be the optimal solution. Applying Lemma~\ref{lem:polytope}, we get $|Y|$ set sequences $\vect{\tilde{S}}_1, \vect{\tilde{S}}_2, \ldots, \vect{\tilde{S}}_{|Y|}$ that satisfy the stated properties. Consider the  two cases below. 
\paragraph{Case 1.} If $g(\vect{\tilde{S}}_{|Y|}) \leq \mu$, we simply define $\vect{D'} = \vect{D}$ and $\pi_i$ be the unique element in $\vect{\tilde{S}}_i \setminus \vect{\tilde{S}}_{i - 1}$. We verify that $\opt_{\mathsf{CO}} = \E_{\types \sim D}[f(y_u^{r^\pi})]$. For each $i$, suppose $\pi_i = (j,t)$. We have
\begin{multline*}
x^{r^\pi}_j(t) \nonp_j(t) = \E_{\types \sim D}[r^{\pi}_j(\types) \cdot \indic[t_j = t] ] = g(\vect{\tilde{S}}_i) - g(\vect{\tilde{S}}_{i - 1}) \\
= \sum_{j \in [n]}\sum_{t' \in S_{i,j}} z^*_j(t') - \sum_{j \in [n]} \sum_{t' \in S_{i - 1,j}} z^*_j(t') = z^*_{j}(t)
\end{multline*}
Consequently,
\begin{multline*}
\E_{\types \sim \vect{D'}}[f(y_u^{r^\pi})] = \sum_{\types} \nonp(\types) \cdot f\left(\sum_{j\in [n]} x^{r^\pi}_j(t_j) \right) = \sum_{\types} \nonp(\types) \cdot f\left(\sum_{j \in [n]} \frac{r_j^*(t_j)}{\nonp_j(t_j)}\right) = \opt_{\mathsf{CO}} \le \E_{\types \sim \vect{D}}[f(y_u^{r})] .
\end{multline*}

\paragraph{Case 2.} If $g(\vect{\tilde{S}}_{|Y|}) > \mu$, consider the following refinement. Let the unique element in $\vect{\tilde{S}}_{|Y|} \setminus \vect{\tilde{S}}_{|Y| - 1}$ be $(j^*, t^*)$. We split the type $t^*$ into two types $t_a$ and $t_b$ and modify the distribution to $D'_{j^*}$ with $\nonp^{D'_{j^*}}(t_a) + \nonp^{D'_{j^*}}(t_b) = \nonp^{D_{j^*}}(t^*)$. We explain below how the values of $\nonp^{D'_{j^*}}(t_a), \nonp^{D'_{j^*}}(t_b)$ are chosen.

For each $1 \leq i \leq |Y| - 1$, let $\vect{\tilde{S}'}_i = \vect{\tilde{S}}_i$. 
Suppose $\vect{\tilde{S}}_{|Y|} = (S_{|Y|,1}, S_{|Y|,2}, \dots, S_{|Y|,n})$. 
Let $\vect{\tilde{S}'}_{|Y|} = (S_{|Y|,1}, S_{|Y|,2}, \dots, (S_{|Y|,j^*} \setminus \{t^*\}) \cup \{t_a\}, \dots, S_{|Y|,n})$. Note that $g(\vect{\tilde{S}'}_{|Y|})$ is monotone in $\nonp^{D'_{j^*}}(t_a)$, we therefore can set the value of $\nonp^{D'_{j^*}}(t_a)$ so that $g(\vect{\tilde{S}'}_{|Y|}) = \mu$. 

For each $j \neq j^*$, let $D'_j = D_j$. Define $\pi_i$ to be the unique element in $\vect{\tilde{S}'}_i \setminus \vect{\tilde{S}'}_{i - 1}$ for $1 \leq i \leq |Y|$. We verify that $\E_{\types \sim \vect{D'}} [f(y_u^{r^{\pi}})] \leq \opt_{\mathrm{CO}}$. 
Similar to the previous case, we first calculate the value of $x^{r^\pi}_j(t) \cdot  \nonp^{\vect{D'}}_j(t)$ for each $\pi_i = (j, t)$. 
We consider the two cases depending on the type $t$:
\begin{itemize}
\item If $t \neq t_a$, we have that 
\begin{multline*}
x^{r^\pi}_j(t) \nonp^{\vect{D'}}_j(t) = \E_{\types \sim \vect{D'}}[r^{\pi}_j(\types) \cdot \indic[t_j = t] ] = g(\vect{\tilde{S}'}_i) - g(\vect{\tilde{S}'}_{i - 1}) \\
= \sum_{j\in[n]} \sum_{t' \in S'_{i,j}} z^*_j(t') - \sum_{j\in [n]} \sum_{t' \in S'_{i - 1,j}} z^*_j(t') = z^*_{j}(t).
\end{multline*}
\item If $t = t_a$, we have that 
\begin{multline*}
x^{r^\pi}_{j^*}(t_a) \nonp^{\vect{D'}}_{j^*}(t_a) = \E_{\types \sim \vect{D'}}[r^{\pi}_{j^*}(\types) \cdot \indic[t_{j^*} = t_a] ] = g(\vect{\tilde{S}'}_i) - g(\vect{\tilde{S}'}_{i - 1}) \\
= \sum_{j \in [n]}\sum_{t' \in S'_{i,j}} z^*_j(t') - \sum_{j\in[n]} \sum_{t' \in S'_{i - 1,j}} z^*_j(t') = z^*_{j}(t^*).	
\end{multline*}
\end{itemize}
Therefore, we have
\begin{multline*}
\E_{\types \sim \vect{D'}}[f(y_u^{r^{\pi}})] = \sum_{\types} \nonp^{\vect{D'}}(\types) \cdot f\left(\sum_{j\in [n]} x_j^{r^{\pi}} (t_j)\right) \\
= \sum_{\types: t_{j^*} \neq t_a} \nonp^{\vect{D}}(\types) \cdot f\left(\sum_{j \in [n]} \frac{z^*_j(t_j)}{\nonp^{\vect{D}}_j(t_j)}\right) +  \sum_{\types: t_{j^*} = t_a} \nonp^{\vect{D'}}(\types) \cdot  f\left(\sum_{j\ne j^*} \frac{z^*_j(t_j)}{\nonp^{\vect{D}}_j(t_j)} + \frac{z^*_{j^*}(t^*)}{\nonp^{\vect{D'}}_{j^*}(t_a)} \right) \\
\le \sum_{\types: t_{j^*} \neq t_a} \nonp^{\vect{D}}(\types) \cdot f\left(\sum_{j \in [n]} \frac{z^*_j(t_j)}{\nonp^{\vect{D}}_j(t_j)}\right) +  \sum_{\types: t_{j^*} = t_a} \nonp^{\vect{D}}(t_{j^*}=t^*,\typesmi[j^*]) \cdot  f\left(\sum_{j\ne j^*} \frac{z^*_j(t_j)}{\nonp^{\vect{D}}_j(t_j)} + \frac{z^*_{j^*}(t^*)}{\nonp^{\vect{D}}_{j^*}(t^*)} \right) = \opt_{\mathsf{CO}}
\end{multline*}

Here, the inequality is equivalent to 
\[
\nonp^{\vect{D'}}_{j^*}(t_a) \cdot f\left(\sum_{j\ne j^*} \frac{z^*_j(t_j)}{\nonp^{\vect{D}}_j(t_j)} + \frac{z^*_{j^*}(t^*)}{\nonp^{\vect{D'}}_{j^*}(t_a)} \right) \leq \nonp^{\vect{D}}_{j^*}(t^*) \cdot  f\left(\sum_{j\ne j^*} \frac{z^*_j(t_j)}{\nonp^{\vect{D}}_j(t_j)} + \frac{z^*_{j^*}(t^*)}{\nonp^{\vect{D}}_{j^*}(t^*)} \right),
\]
which follows from the fact that $\nonp^{\vect{D'}}_{j^*}(t_a) \leq \nonp^{\vect{D}}_{j^*}(t^*)$ and the following inequality: for any concave function $f$ with $f(c) \geq 0$ and $a \leq b$, we have 
\[
a \cdot f \left(c + \frac{x}{a}\right) \leq (b - a) \cdot f(c) + a \cdot f\left(c + \frac{x}{a}\right) = b \cdot \left(\frac{b - a}{b} \cdot f(c) + \frac{a}{b} \cdot f\left(c + \frac{x}{a} \right) \right) \leq b \cdot f\left(c + \frac{x}{b}\right).
\]
This concludes the proof of $\E_{\types \sim \vect{D'}} [f(y_u^{r^{\pi}})] \leq \opt_{\mathsf{CO}} \le \E_{\types \sim \vect{D}}[f(y_u^{r})]$.
\end{proofof}
 \subsubsection{Arrival Distributions}
\label{subsec:outer}

By Lemma~\ref{lem:noniidworstper}, for any distribution $\vect{D}$ and selection rule $r$, there exists a refinement $\vect{D'}$ and a permutation rule $r^{\pi}$ which gives worse ratio for the optimization problem \eqref{opt-prob}. Hence it is without loss of generality restricting ourselves to permutation rules.
We explicitly find the worst-case arrivals under permutation rules in this section. Specifically, we construct a splitting operation of online vertices, and prove that the value $\E[f(y)]$ always decreases after each split, given $f$ is concave.\footnote{In the proof of Lemma~\ref{lem:noniidworstper}, we did a similar operation which split the type $t^* \in T_{j^*}$ of an online vertex $v_{j^*}$. It worths pointing out the difference that here the number of online vertex is increased after splitting, while in Lemma~\ref{lem:noniidworstper}, the number of vertices is fixed.} Before formally state our main lemma, we set up necessary notations and define the operation. 

\paragraph{Notations.} 
We shall fix an offline vertex $u\in L$ in the analysis. For notation simplicity, we drop the subscript of $u$. We focus on instances $\mathcal{I} = (\pi, R, \vect{D})$ that are parameterized by a permutation $\pi$, a set of all online vertices $R$ and corresponding distributions $\vect{D} = (D_1, D_2, \dots, D_{|R|})$. We use $T_j,\nonp_j$ to denote the support and the probability mass function of distribution $D_j$ for each $j$. 

\paragraph{Splitting an Online Vertex.} 	Given an instance $\mathcal{I} = (\pi, R, \vect{D})$, a vertex $v_j \in R$, and an arbitrary $\varepsilon > 0$. Consider the following splitting operation.

Let $T_j = \{a_1, a_2, \dots, a_k, \varnothing\}$ be the support of vertex $v_j$, where we use $\varnothing$ to represent the types that are not involved in the support of $\pi$. I.e., those types would not be selected by the permutation rule $r^{\pi}$ and we can safely merge them to a single type that we denote by $\varnothing$. Suppose the types $\{a_1,a_2,\dots,a_k\}$ are in ascending order with respect to $\pi$, i.e. $(\nj, a_1) <_{\pi} (\nj, a_2) <_{\pi} \dots <_{\pi} (\nj, a_{k})$.
We split the vertex $v_j$ into two vertices $v_{j'}$ and $v_{j''}$ with the following distributions:
\begin{align*}
t_{j'}=
\begin{cases}
a_1'  & \text{w.p. } \frac{\nonp_{\nj}(a_1) - \varepsilon}{1 -\varepsilon}\\
a_\ell \quad \forall 2\le \ell \le k, & \text{w.p. } \frac{\nonp_{\nj}(a_\ell)}{1-\varepsilon} \\
\varnothing & \text{w.p. } \frac{\nonp_{\nj}(\varnothing)}{1-\varepsilon}
\end{cases}
\quad\text{and}\quad
t_{j''} = 
\begin{cases}
a_1'' & \text{w.p. } \varepsilon \\
\varnothing & \text{w.p. } 1 - \varepsilon
\end{cases}
\end{align*}
We denote the two distributions by $D_{j'},D_{j''}$ and they are well-defined for $\varepsilon \in (0, \nonp_j(a_1)]$. 

Consider a new instance with online vertices $R' = R \setminus \{v_{\nj}\} \cup \{v_{j'}, v_{j''}\}$, distributions $\vect{D'} = (\vect{D}_{\text{-}j}, D_{j'}, D_{j''})$, and the permutation $\pi'$ defined below.
Suppose $$\pi  = (\cdots, (\nj, a_1), \cdots, (\nj, a_2),\cdots, (\nj,a_k), \cdots).$$ 
Then $\pi'$ is defined as $$\pi' = (\cdots, (j', a_1'), (j'', a_1''),\cdots, (j', a_2), \cdots, (j', a_k)).$$ 
That is, we change $(j, a_i)$ to $(j', a_i)$ for every $i \ge 2$ and substitute $(j',a_1'), (j'',a_1'')$ for $(j, a_1)$. The order of other types remains unchanged. 

Now we are ready to state the main lemma of this section. It states that the instance $\mathcal{I'}$ after an arbitrary split is always worse than the original instance $\mathcal{I}$. Therefore, the worst competitive ratio is achieved in the limit case when we keep splitting the instance. The implication of the lemma shall be explained in the next section.

\begin{lemma}
\label{lem:split}
	  For any instance $\mathcal{I} = (\pi, R, \vect{D})$, let $\mathcal{I'} = \left(\pi', R', \vect{D'}\right)$ be the instance after splitting $v_{\nj} \in R$. For any concave function $f$, we have 
\begin{align}\frac{\E_{\types \sim \vect{D}}\left[f(y^{r^{\pi}})\right]}{\E_{\types \sim \vect{D}}\left[y^{r^{\pi}}\right]} \ge \frac{\E_{\types \sim \vect{D'}}[f(y^{r^{\pi'}})]}{\E_{\types \sim \vect{D'}}[y^{r^{\pi'}}]}. \label{eq:split}\end{align}
\end{lemma}
\begin{proof}
For notation simplicity, within the proof of the lemma, we use $x_i(t), x_i'(t)$ to denote $x_i^{r^\pi}(t_i), x_i^{r^{\pi'}}(t)$, and $y(\types), y'(\types)$ to denote $y^{r^{\pi}}(\types), y^{r^{\pi'}}(\types)$. 

We first observe that the mean of $y,y'$ are equal:

\begin{multline*}
\E_{\types \sim \vect{D}}[y] = 1 - \prod_{i: v_i \in R}  \nonp^{D_i}(\varnothing) = 1 - \prod_{i \ne j: v_i \in R} \nonp^{D_i'}(\varnothing) \cdot \nonp^{D_j}(\varnothing) \\
= 1 - \prod_{i \ne j: v_i \in R} \nonp^{D_i'}(\varnothing) \cdot \nonp^{D_{j'}}(\varnothing) \cdot \nonp^{D_{j''}}(\varnothing)= \E_{\types \sim \vect{D'}}[y'],
\end{multline*}
where the third equality follows from the definition of $D_{j'}$ and $D_{j''}$.

Hence, it suffices to prove that $\E_{\types \sim \vect{D}}\left[f(y)\right] \geq \E_{\types \sim \vect{D'}}\left[f(y')\right]$. Note we have 
\begin{equation}
\label{eq:y_and_y'}
y(\types) = \sum_{\substack{v_{\ell} \in R \\ \ell \neq j}} x_\ell(t_\ell) + x_j(t_j), \quad y'(\types) = \sum_{\substack{v_{\ell} \in R' \\ \ell \neq j',j''}} x'_{\ell} (t_{\ell}) + x'_{j'}(t_{j'}) + x'_{j''}(t_{j''}). 
\end{equation}

It is straightforward to verify that for every $\ell \neq j$, the value of $x_{\ell}(t) = x'_\ell(t)$. This follows by the definition of the permutation rule and we omit the tedious proof. 

Moreover, our splitting operation guarantees that for arbitrary constant of $c$:
\begin{equation}
\label{eq:noniidaim}
\E_{t_\nj}\left[f\left(c + x_\nj(t_\nj)\right)\right] \ge \E_{t_{j'},t_{j''}}\left[f\left(c + x'_{j'}(t_{j'}) + x'_{j''}(t_{j''})\right)\right],
\end{equation}
whose proof is deferred to the end of the section. 
Combining this inequality with the fact that $x_\ell(t) = x'_\ell(t)$ for $\ell \ne j$, we have that
\begin{multline*}
\E_{\types \sim \vect{D}}[f(y(\types))] = \E_{\types \sim \vect{D}}\left[f\left(\sum_{\substack{v_{\ell} \in R \\ \ell \neq j}} x_\ell(t_\ell) + x_j(t_j)\right)\right] \\
\ge \E_{\types \sim \vect{D'}}\left[ f\left(\sum_{\substack{v_{\ell} \in R' \\ \ell \neq j',j''}} x'_{\ell} (t_{\ell}) + x'_{j'}(t_{j'}) + x'_{j''}(t_{j''})\right)\right] = \E_{\types \sim \vect{D'}}[f(y'(\types)],
\end{multline*}
that concludes the proof of the lemma.
\end{proof}

\begin{proofof}{\eqref{eq:noniidaim}}
We define $g(x) = f(c + x)$ which is also a concave function. We first expand the left hand side of \eqref{eq:noniidaim} by definition. Suppose $T_j = \{a_1, a_2, \dots, a_k, \varnothing\}$. We have
\begin{align}
    \begin{split}\label{eq:noniidLHS}
        \E_{t_\nj\sim D_\nj}\left[g\left(x_\nj(t_\nj)\right)\right] = \sum_{i=1}^k \nonp^{D_j}(a_i)\cdot g(x_\nj(a_i)) + \left( 1 - \sum_{i=1}^k\nonp^{D_j}(a_i)\right)\cdot g(0).
    \end{split}
\end{align}
We also expand the right hand side. Suppose $T_{j'} = \{a'_1, a_2, \dots, a_k, \varnothing\}$. For simplicity, we define $\tilde a_1 = a'_1$, and $\tilde a_i = a_i$ for $i > 1$. 

$$g(x'_{j'}(t_{j'}) + x'_{j''}(t_{j''})) = \begin{cases}
g(x'_{j'}(\tilde{a}_i) + x'_{j''}(a_1'')) & \text{when $t_{j'} = \tilde{a}_i, t_{j''} = a_1''$ w.p. $\varepsilon\cdot \nonp^{D_{j'}}(\tilde a_i)$} \\ g(x'_{j'}(\tilde{a}_i)) & \text{when $t_{j'} = \tilde{a}_i, t_{j''} = \varnothing$ w.p. $(1-\varepsilon)\cdot \nonp^{D_{j'}} (\tilde a_i)$} \\ g(x'_{j''}(a''_1)) & \text{when $t_{j'}=\varnothing, t_{j''} = a''_1$ w.p. $\left(1 - \sum_{i=1}^k \nonp^{D_{j'}} (\tilde a_i)\right)\cdot \varepsilon$}\\g(0) & \text{when $t_{j'}=\varnothing, t_{j''} = \varnothing $ w.p. $\left(1 - \sum_{i=1}^k \nonp^{D_{j'}}(\tilde{a}_i)\right)\cdot (1 - \varepsilon)$}
\end{cases}$$

And we get the following equation:

\begin{align}
\label{eq:noniidRHS}
    \begin{split}
       \E_{\substack{t_{j'} \sim D_{j'} \\ t_{j''} \sim D_{j''}}} \left[g\left(x_{j'}'(t_{j'})+x_{j''}'(t_{j''})\right)\right]
         &=  \sum_{i=1}^k \left(\varepsilon \nonp^{D_{j'}}(\tilde a_i) \cdot  g\left(  x_{j'}'(\tilde a_i)+x_{j''}'(a_1'') \right) + (1 - \varepsilon) \nonp^{D_{j'}}(\tilde a_i) \cdot  g\left(  x_{j'}'(\tilde a_i)\right)\right) \\
            & \quad\  + \left(1 - \sum_{i=1}^k \nonp^{D_{j'}}(\tilde a_i) \right)\cdot \varepsilon \cdot g\left( x_{j''}'(a_1'')\right) + \left(1 - \sum_{i=1}^k \nonp^{D_{j'}}(\tilde a_i)\right)\cdot (1-\varepsilon)\cdot g(0) 
    \end{split}
\end{align}

We now handle the terms of \eqref{eq:noniidRHS} with $i > 1$ . 

\begin{claim} \label{claim:igeq1}
For any concave function $g$, suppose $i > 1$ (which means $\tilde a_i = a_i$),
    \begin{multline}    
    \label{eq:noniidai}
    \varepsilon \nonp^{D_{j'}} (a_i) \cdot  g\left(  x_{j'}'(a_i)+x_{j''}'(a_1'') \right) + (1 - \varepsilon) \nonp^{D_{j'}} (a_i) \cdot  g\left(  x_{j'}'(a_i)\right)\\
    \le \nonp^{D_j}(a_i)\cdot g(x_j(a_i)) + \varepsilon \nonp^{D_{j'}}(a_i) g(x_{j''}'(a_1''))
    \end{multline}
\end{claim}

\begin{proof}
Suppose $d > c, b> a$, and $c > a$, from the concavity of $g$, it follows that $$g(d) - g(c) \le \frac{d-c}{b-a} \left(g(b) - g(a)\right).$$
    
Notice that by our definition of splitting, $ x_{j'}'(a_i) = (1 - \varepsilon) \cdot x_j(a_i) $ for $i > 1$. 
We take $ d = x_{j'}'(a_i)+x_{j''}'(a_1''), c = x_{j''}'(a_1''), b = x_{j}(a_i), a= x_{j'}'(a_i) = (1-\varepsilon)\cdot x_{j}(a_i)$ and get
$$g(x_{j'}'(a_i)+x_{j''}'(a_1'')) - g( x_{j''}'(a_1''))\ge \frac{1 - \varepsilon}{\varepsilon} \left(g(x_{j}(a_i)) - g(x_{j'}'(a_i))\right).$$

Rewrite the formula above and multiply both sides by $\nonp^{D_{j'}}(a_i)$, we then get \eqref{eq:noniidai}.
\end{proof}

Now we are left with the case when $i = 1$.

\begin{claim}
\begin{multline}
\label{eq:noniida1}
\frac{\nonp^{D_j}(a_1) - \varepsilon}{1 - \varepsilon}\cdot \varepsilon\cdot g\left(  x_{j'}'(a'_1)+x_{j''}'(a''_1) \right)+   \left(\nonp^{D_j}(a_1) - \varepsilon\right) \cdot g\left( x_{j'}'(a'_1) \right) \\
+\left(1 - \frac{\nonp^{D_j}(a_1) - \varepsilon}{1 - \varepsilon}\right)\cdot \varepsilon \cdot g\left( x_{j''}'(a_1'')\right) \leq \nonp^{D_j}(a_1) \cdot g( x_j(a_1))
\end{multline}
\end{claim}

\begin{proof}
By definition, we have the following facts,
\begin{multline}
\label{eq:noniida1E}
\frac{\nonp^{D_j}(a_1) - \varepsilon}{1 - \varepsilon}\cdot \varepsilon\cdot \left(  x_{j'}'(a'_1)+x_{j''}'(a_1'') \right)+   \left(\nonp^{D_j}(a_1) - \varepsilon\right) \cdot  x_{j'}'(a'_1)  \\
+\left(1 - \frac{\nonp^{D_j}(a_1) - \varepsilon}{1 - \varepsilon}\right)\cdot \varepsilon \cdot x_{j''}'(a_1'') = \nonp^{D_j}(a_1) \cdot x_j(a_1)
\end{multline}
and 
\begin{equation}
\label{eq:noniida1Prob}
\frac{\nonp^{D_j}(a_1) - \varepsilon}{1 - \varepsilon}\cdot \varepsilon +   \left(\nonp^{D_j}(a_1) - \varepsilon\right) +\left(1 - \frac{\nonp^{D_j}(a_1) - \varepsilon}{1 - \varepsilon}\right)\cdot \varepsilon  = \nonp^{D_j}(a_1).
\end{equation}

Then, inequality~\eqref{eq:noniida1} directly follows from \eqref{eq:noniida1E} , \eqref{eq:noniida1Prob} and Jensen's inequality.
\end{proof}

Finally, by summing up \eqref{eq:noniida1} and \eqref{eq:noniidai} from $2$ to $k$, we have
\begin{align}\label{eq:noniidresult}
    \begin{split}
        &\nonp^{D_j}(a_1) \cdot g\left( x_j(a_1)\right) + \sum_{i=2}^k\left(\nonp^{D_j}(a_i)\cdot g(x_j(a_i)) + \varepsilon \nonp^{D_{j'}}(a_i) g(x_{j''}'(a_1''))\right)\\
        \ge & \frac{\nonp^{D_j}(a_1) - \varepsilon}{1 - \varepsilon}\cdot \varepsilon\cdot g\left(  x_{j'}'(a_1')+x_{j''}'(a_1'') \right)+   \left(\nonp^{D_j}(a_1) - \varepsilon\right) \cdot g\left( x_{j'}'(a_1') \right)+\left(1 - \frac{\nonp^{D_j}(a_1) - \varepsilon}{1 - \varepsilon}\right)\cdot \varepsilon \cdot g\left( x_{j''}'(a_1'')\right)\\
        & + \sum_{i=2}^k \left(\varepsilon \nonp^{D_{j'}}(a_i) \cdot  g\left(  x_{j'}'(a_i)+x_{j''}'(a_1'') \right) + (1 - \varepsilon) \nonp^{D_{j'}}(a_i) \cdot  g\left(  x_{j'}'(a_i)\right)\right) 
    \end{split}\end{align}

Rearranging the formula above, we find that \eqref{eq:noniidresult} is equivalent to \eqref{eq:noniidaim} by \eqref{eq:noniidLHS} and \eqref{eq:noniidRHS}.
\end{proofof}
 \subsubsection{The Worst-case Distribution for Independent Estimator} \label{appendix:worst-case-distribution}

Recall that the competitive ratio $\Gamma$ is lower bounded by the following optimization problem:
\begin{equation}\label{eq:noniidtarget}
	\inf_{D_1, D_2, \dots, D_n} \inf_{r} \frac{\E[f(y_u^r)]}{\E[y_u^r]}
\end{equation}

In Section~\ref{subsec:worst-rule}, we show that permutation rules are the worst among all possible selection rules. In Section~\ref{subsec:outer}, we show that splitting the instance while preserving the mean $\mu = \E[y_u^r]$ would only decreases the ratio.
Given an instance $\mathcal{I}=(\pi,R,\vect{D})$, by applying the splitting operation multiple times to every online vertex $v_j \in R$, we first convert the distributions $\vect{D}$ into ``Bernoulli'' distributions, i.e., each online vertex $v_j$ either has type $t_j$, or has type $\varnothing$. This can be achieved by setting $\varepsilon = \nonp_j(a_1)$ in the splitting operation.
Next, we fix an sufficiently small $\varepsilon$ and continue splitting the instance until the probability of each vertex being realized (i.e., its type is not $\varnothing$) is at most $\varepsilon$.
According to Lemma~\ref{lem:split}, the ratio monotonically decreases during the above procedure.
Finally, we can remove all vertices with realized probability smaller than $\varepsilon$. The impact of such change can be made vanishingly small as $\varepsilon \to 0$. At the end, we get a regularized instance defined as the following.

\begin{definition}[Worst-case distribution for fixed $\mu$]
	Fixing $\mu = \E[y_u^r]$, $\mathcal{W}_n$ is an instance with $n$ online vertices and a permutation selection rule $r^{\pi}$. For the $i$-th arrived vertex $v_i$, it has probability $\epsilon$ to be type $t_i$, and otherwise it has no edges. Here $\epsilon$ satisfies that $1 - (1-\epsilon)^n = \mu$. The permutation selection rule $r^{\pi}$ always selects the type $t_k$ with the \textbf{maximal} subscript.	
\end{definition}

\begin{proofof}{Theorem~\ref{thm:noniid}}
We first give the competitive ratio for both fractional and integral non-i.i.d. matching and then show that independent estimator reaches the best possible ratio among subset-resampling estimators.
    
According to the discussion above, the infimum of the optimization problem is achieved in the limit case of $\mathcal{W}_n$ when $n\rightarrow \infty$ for arbitrary fixed number of $\mu$. We use computer assistance to get a numerical result. We leave the details of our implementations in Section~\ref{sec:exper}, and plot the numerical results in the following figure:
\begin{figure}[H]
    \centering
    \includegraphics[scale=0.65]{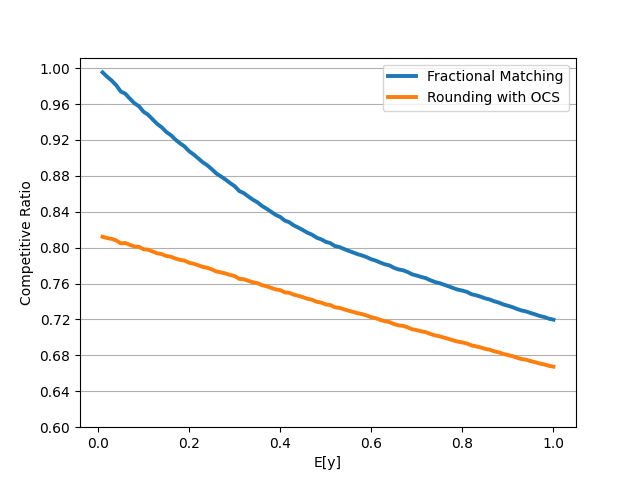}
    \caption{Competitive ratio for different $\mu$}
    \label{fig:noniidexperiments}
\end{figure}

We conclude that the competitive ratios for fractional and integral matching are respectively $0.718$ and $0.666$.
Finally, notice that in the definition of $\mathcal{W}_n$, we specify the permutation $\pi$ to be the reverse order of the arrival of vertices. 
Observe that for such instances, all subset-resampling estimators described in Section~\ref{sec:iid} are equivalent to the independent estimator, since the presence of types before vertex $v_i$ shall not influence the selection of type $t_i$ under such permutation rule. 
In other words, for this restricted family of algorithms, they all achieve the same competitive ratio as the independent estimator for such instances. 
As we have proved this ratio is the worst case ratio for the independent estimator, we conclude that the independent estimator is optimal among all subset-resampling algorithms.
\end{proofof}
 
\subsection{Experiments in Non I.I.D. Arrivals}\label{sec:exper}

In this section we will present the details of our experiments in non-i.i.d. arrivals. The codes which generate all the numerical results presented in this section are open-sourced at this \href{https://github.com/wuzhuangtai00/STOC22-Fractional-Stochastic-Matching}{link}. 

As we have shown in Section~\ref{appendix:worst-case-distribution}, the worst case distribution $\mathcal{W}_n$ is as following: Given the parameter $\mu$ which is the probability that the selection function $f^\pi$ will select one online vertex, there are $n$ online vertices arrive in order. For the $i$-th arrived vertex, it has a probability $p$ to have type $t_i'$, and otherwise it has an empty type that will never be selected by the selection function $f^{\pi}$. The selection function $f^{\pi}$ will always try to select a type with the maximal index, and so for the $j$-th arrived online vertex, the independent estimator $x_{j}^{f^{\pi}}$ is the probability that no vertex with an index larger than $j$ has a non-empty set, which is $(1 - p)^{n-j}$. Therefore, we could easily draw a sample of $$y = \sum_{j=1}^n x_j^{f^{\pi}}(t_i)$$ from the worst case distribution given $n$ and $\mu$ by the following algorithm:

\begin{algorithm}[H]
\label{alg:noniidexperiment}
  \SetAlgoLined
    \Parameter{$\mu$: The probability that selection function selects one vertex.  \\ $n$: Number of offline vertices. }
    
  \KwResult{ $y \leftarrow \sum_{j = 1}^n x_j^{f}(t_j)$ : A sample of sum of independent estimators drawn from worst case distribution.}

    {$p \leftarrow 1 - (1 - \mu) ^ {\frac{1}{n}}$\;}
  \tcp{$p$ is the number satisfying $1 - (1 - p) ^ n = \mu$}
  
  {$y \leftarrow 0$\;}
  \For{$i \leftarrow 1 \text{ to } n$}{
    $x_i \leftarrow (1 - p) ^ {n - i}$, $\text{with probability } p$, and $x_i \leftarrow 0$ otherwise\\
    $y \leftarrow y + x_i$\;
}
  \caption{Draw one sample from the worst case distribution}
\end{algorithm}

Notice that the ratio of fractional and integral matching are  the minimum of the following optimization problem when $f(x)$ respectively equals to $\min(x, 1)$ and $f(x) = 1 - \exp\left(-x-\frac{x^2}{2}-\frac{4-2\sqrt{3}}3x^3\right)$:
\[
\inf_{D_1, D_2, \dots, D_n} \inf_{r} \frac{\E[f(y_u^r)]}{\E[y_u^r]}.
\]
We discretize $\mu$ to be multiples of $0.01$, and set $n$ to be $1000$ for our numerical experiment. Due to the concentration bound for summation of independent random variable, our experiment leads to a reasonably accurate estimation of the real performance of our algorithm.  
\bibliographystyle{plain}
\bibliography{matching}

\appendix

\section{From Second Moment to Competitive Ratio} \label{appendix:variance-LP}

This section is devoted to the proof of Lemma \ref{lem:variance-LP}. 

\begin{reminder}{Lemma \ref{lem:variance-LP}}
Suppose a non-negative random variable $y$ satisfies that $\E[y] = \mu$ and $\E[y^2] \leq \gamma$. 
\begin{enumerate}[label=(\alph*)]
	\item If $\mu \in [0,1]$ and $\gamma = \mu + \frac{1}{2}\mu^2 $, then $\E[\min(y,1)] \geq 0.646 \cdot \E[y]$;
	\item If $\mu \in [0,1]$ and $\gamma = \mu + \frac{1}{2}\mu^2 $, then $\E[p(y)] \geq 0.634 \cdot \E[y]$;
	\item If $\mu \in [0,1]$ and $\gamma = 1.05771\mu + 0.231\mu^2$, then $\E[\min(y,1)] \geq 0.731 \cdot \E[y]$.
	\item If $\mu \in [0,1]$ and $\gamma = 1.05771\mu + 0.231\mu^2$, then $\E[p(y)] \geq 0.704 \cdot \E[y]$.
\end{enumerate}
\end{reminder}

\paragraph*{Proof intuition}

Informally speaking, let $q(y)$ be the p.d.f. of $y$ with support $[0,+\infty)$. We are solving the following LP  \footnote{This formulation is for intuition only. For example, it does not take care of mass point of the distribution of $y$.} where $f(y)$ is either $\min(y,1)$ or $p(y)$.

\begin{align*}
\mathrm{minimize} \ \ &\int_{0}^{\infty} f(y) q(y) \,dy \\	
\mathrm{s.t.} \ \ & \int_0^\infty y^2 q(y) \, dy \leq \gamma \\
				  & \int_0^\infty y q(y) \, dy = \mu\\
				  & \int_0^\infty q(y) \, dy = 1\\
				  & \forall y \geq 0 \ , \ q(y) \geq 0 
\end{align*}

 Then its dual problem is the following.
\begin{align*}
\mathrm{maximize} \ \ & a\gamma + b\mu + d\\
\mathrm{s.t.} \ \ & \forall y \geq 0 \ , \ a y^2+by+d\leq f(y) \\
&a \leq 0
\end{align*}

Intuitively speaking, the dual problem hints us that we can use a quadratic function $ay^2 + by + d$ with $a \leq 0$ to lower bound $1 - \exp(-y - y^2 /2 - cy^3)$, and then plug in our bound for $\E[y]$ and $\E[y^2]$. This naturally gives our desired lower bound for the oringal LP.

\begin{proofof}{Lemma \ref{lem:variance-LP}}

\paragraph*{(a)}

Since we have to prove the lemma for all $0 \leq \mu \leq 1$, intuitively, we need to lower bound the function $\E[\min(y,1)]$ with different quadratic functions for different $\mu$. It turns out that two different quadratic functions sufffice. 

When $\mu \leq 0.35$, we choose quadratic function $a_1y^2 +b_1y+d_1$ with $a_1 = -0.3, b_1 = 1, d_1 = 0$. It lower bounds $\min(y,1)$, and therefore $$\E[\min(y,1)] \geq a_1 \E[y^2] + b_1 \E[y] + d_1 \geq a_1 \gamma + b_1 \mu + d_1.$$

Plug in $\gamma = \mu + \frac{1}{2}\mu^2$, this proves $\E[\min(y,1)] \geq 0.646\mu$ when $0 \leq \mu \leq 0.35$.

When $\mu > 0.35$, we take $a_2y^2 +b_2y+d_2$ with $a_2 = -0.35368, b_2 = 1.20735, d_2 = -0.03040$. It lower bounds $\min(y,1)$, and therefore $$\E[\min(y,1)] \geq a_2 \E[y^2] + b_2 \E[y] + d_2 \geq a_2 \gamma + b_2 \mu + d_2.$$

Plug in $\gamma = \mu + \frac{1}{2}\mu^2$, this proves $\E[\min(y,1)] \geq 0.646\mu$ when $0.35 < \mu \leq 1$.

\paragraph*{(b)}

We follow essentially the same approach as (a). 

When $\mu \leq 0.4$, we still choose quadratic function $a_1y^2 +b_1y+d_1$ with $a_1 = -0.3, b_1 = 1, d_1 = 0$. 

When $\mu > 0.4$, we choose $a_2y^2 +b_2 y+d_2$ with $a_2= -0.3099, b_2 = 1.1108, d_2 = -0.0113$. 

These two functions both lower bound $\E[p(y)]$, and together with $\gamma = \mu + \frac{1}{2}\mu^2$, following same argument as (a), it proves $\E[p(y)] \geq 0.634\mu$

\paragraph*{(c)}

In this case, we need more quadratic functions. 

When $\mu \leq 0.07$, we choose quadratic function $a_1y^2 +b_1y+d_1$ with $a_1 = -0.25, b_1 = 1, d_1 = 0$. 

When $0.07 \leq \mu < 0.21$, we choose $a_2 y^2 + b_2 y + d_2$ with $a_2 = -0.2622, b_2 = 1.0242, d_2 = -0.0006$.

When $0.21 \leq \mu < 0.37$, we choose $a_3 y^2 + b_3 y + d_3$ with $a_3 = -0.2907, b_3 = 1.0813, d_3 = -0.0057$.

When $0.37 \leq \mu < 0.64$, we choose $a_4 y^2 + b_4 y + d_4$ with $a_4 = -0.3265, b_4 = 1.1528, d_4 = -0.0179$.

When $0.64 \leq \mu < 0.78$, we choose $a_5 y^2 + b_5 y + d_5$ with $a_5 = -0.3714, b_5 = 1.2427, d_5 = -0.0397$.

When $0.78 \leq \mu < 0.91$, we choose $a_6 y^2 + b_6 y + d_6$ with $a_6 = -0.4295, b_6 = 1.3589, d_6 = -0.0750$. 

When $\mu \geq 0.91$, we choose $a_7 y^2 + b_7 y + d_7$ with $a_7 =  -0.4654, b_7 = 1.4307, d_7 = -0.0997$.

These functions all lower bound $\E[\min(y,1)]$, and together with $\gamma = 1.05771\mu + 0.231\mu^2$, following same argument as (a), it proves $\E[\min(y,1)] \geq 0.731\mu$.

\paragraph*{(d)}

When $\mu \leq 0.4$, we choose quadratic function $a_1y^2 +b_1y+d_1$ with $a_1 = -0.252, b_1 = 1, d_1 = 0$. 

When $\mu > 0.4$, we take $a_2y^2 +b_2y+d_2$ with $a_2 = -0.347711, b_2 = 1.180665, d_2 = -0.028471$.

These two functions both lower bound $\E[p(y)]$, and together with $\gamma = 1.05771\mu + 0.231\mu^2$, following same argument as (a), it proves $\E[p(y)] \geq 0.704\mu$.

\end{proofof} 
\end{document}